\pdfoutput=1 
\documentclass[11pt]{article}

\usepackage[utf8]{inputenc}
\usepackage[pagebackref, hidelinks, bookmarks, hypertexnames=false]{hyperref}
\usepackage{physics}
\usepackage{amssymb, amsmath, amsthm}
\usepackage{authblk}
\usepackage{mathtools}
\usepackage[noabbrev, nameinlink, capitalise]{cleveref}
\usepackage{dsfont}
\usepackage{comment}
\usepackage[margin=1in]{geometry}

\newtheorem{theorem}{Theorem}[section]
\newtheorem{lemma}[theorem]{Lemma}
\newtheorem{proposition}[theorem]{Proposition}
\newtheorem{corollary}[theorem]{Corollary}
\newtheorem{remark}[theorem]{Remark}

\theoremstyle{definition}

\newtheorem*{theorem*}{Theorem}
\newtheorem*{corollary*}{Corollary}
\newtheorem*{acknowledgements}{Acknowledgements}

\DeclareMathOperator{\C}{\mathbb{C}}

\DeclareMathOperator{\Z}{\mathbb{Z}}

\DeclareMathOperator{\cA}{\mathcal{A}}
\DeclareMathOperator{\cB}{\mathcal{B}}
\DeclareMathOperator{\cC}{\mathcal{C}}
\DeclareMathOperator{\cE}{\mathcal{E}}

\DeclareMathOperator{\cG}{\mathcal{G}}
\DeclareMathOperator{\cH}{\mathcal{H}}

\DeclareMathOperator{\cN}{\mathcal{N}}

\DeclareMathOperator{\cP}{\mathcal{P}}
\DeclareMathOperator{\cR}{\mathcal{R}}
\DeclareMathOperator{\cS}{\mathcal{S}}
\DeclareMathOperator{\cT}{\mathcal{T}}

\DeclareMathOperator{\1}{\mathds{1}}
\DeclareMathOperator{\poly}{poly}

\newcommand{\IdState}{\pi}
\newcommand{\dotarg}{\,\cdot\,}
\newcommand{\DBS}[2]{\widehat{D}( #1\mathbin{\|} #2)}
\newcommand{\HBS}[3]{\widehat{H}_{#1}( #2\mathbin{|} #3)}
\newcommand{\IBS}[3]{\widehat{I}_{#1}( #2: #3)}
\newcommand{\IosBS}[4]{{\widehat{I}}^{\mathrm{os}}_{#1}(#2\mathbin{;}#3\mathbin{|}#4)}
\newcommand{\ItsBS}[4]{{\widehat{I}}^{\mathrm{ts}}_{#1}(#2\mathbin{;}#3\mathbin{|}#4)}
\newcommand{\IrevBS}[4]{{\widehat{I}}^{\mathrm{rev}}_{#1}(#2\mathbin{;}#3\mathbin{|}#4)}

\DeclareMathOperator{\diam}{diam}

\newcommand{\eps}{\varepsilon}

\newcommand{\crit}{\beta_c(\lambda)}
\newcommand{\tcrit}{\widetilde\beta_c(\lambda)}

\newcommand{\vertiii}[1]{{\left\vert\kern-0.25ex\left\vert\kern-0.25ex\left\vert #1 
    \right\vert\kern-0.25ex\right\vert\kern-0.25ex\right\vert}}

\usepackage[dvipsnames]{xcolor}
\definecolor{coolblue}{RGB}{0,51,102}
\definecolor{lightblue}{RGB}{102,210,255}
\definecolor{lightpurple}{RGB}{140,30,255}
\definecolor{lightpink}{RGB}{204,0,204}
\definecolor{midblue}{RGB}{0,102,204}
\definecolor{midpink}{RGB}{153,0,153}
\definecolor{darkblue}{RGB}{0,0,153}
\definecolor{cyan}{RGB}{0,204,204}
\definecolor{lightgreen}{RGB}{0,255,128}
\definecolor{midgreen}{RGB}{0,204,0}
\definecolor{midyellow}{RGB}{204,204,0}
\definecolor{darkyellow}{RGB}{153,153,0}
\definecolor{darkpurple}{RGB}{102,0,102}
\definecolor{orange}{RGB}{255,153,51}
\definecolor{darkred}{RGB}{153,0,76}
\definecolor{lightyellow}{RGB}{255,255,153}
\definecolor{lightred}{RGB}{255,153,153}

\usepackage{tikz}
\usepackage{ifthen}
\usepackage{pgfplots}
\pgfplotsset{compat=1.9}
\usetikzlibrary{shapes,arrows}
\usetikzlibrary{positioning}
\usetikzlibrary{shapes.geometric}

\def\Block[#1,#2,#3,#4]{

\def\r{0.3};

\ifthenelse{\NOT #4=0}{
\fill [#2] (-0.5,-0.5) rectangle ({#1-0.5},0.5);
}

\foreach \n in {1,...,#1}{ 

\shade[shading=ball, ball color=darkred] ({\n-1},0) circle (\r);

}

}

\makeatletter
\def\@fnsymbol#1{%
  \ifcase#1 
    \or \textdagger 
    \or \textdaggerdbl 
    \or \S 
    \or \P 
    \or \textbullet 
    \or \ddagger 
    \else \# 
  \fi}
\makeatother

\begin{document}

\title{Conditional Independence of 1D Gibbs States with Applications to Efficient Learning}
\author[1]{\'Alvaro M. Alhambra\thanks{alvaro.alhambra@csic.es}}
\author[2,3]{{\'A}ngela Capel\thanks{angela.capel@uni-tuebingen.de}}
\author[2]{Paul Gondolf$*$\thanks{paul.gondolf@uni-tuebingen.de}}
\author[2,4]{Alberto Ruiz-de-Alarc\'on\thanks{albertoruizdealarcon@ucm.es}}
\author[3]{Samuel O. Scalet\thanks{sos25@cam.ac.uk\\$*$ corresponding author}}
\affil[1]{\small Instituto de F\'isica Te\'orica UAM/CSIC, C/ Nicol\'as Cabrera 13-15, Cantoblanco, 28049 Madrid, Spain}
\affil[2]{Fachbereich Mathematik, Universit\"at T\"ubingen, 72076 T\"ubingen, Germany}
\affil[3]{Department of Applied Mathematics and Theoretical Physics, University of Cambridge, United Kingdom}
\affil[4]{Departamento de An\'alisis Matem\'atico y Matem\'atica Aplicada, Universidad Complutense de Madrid, Plaza de las Ciencias 3, 28040 Madrid, Spain.}

\date{\today}

\maketitle

\begin{abstract}

  We show that spin chains in thermal equilibrium have a correlation structure in which individual regions are strongly correlated at most with their near vicinity. We quantify this with alternative notions of the conditional mutual information, defined through the so-called Belavkin-Staszewski relative entropy. We prove that these measures decay superexponentially at every positive temperature, under the assumption that the spin chain Hamiltonian is translation-invariant. Using a recovery map associated with these measures, we sequentially construct tensor network approximations in terms of marginals of small (sublogarithmic) size. As a main application, we show that classical representations of the states can be learned efficiently from local measurements with a polynomial sample complexity. We also prove an approximate factorization condition for the purity of the entire Gibbs state, which implies that it can be efficiently estimated to a small multiplicative error from a small number of local measurements. The results extend from strictly local to exponentially-decaying interactions above a threshold temperature, albeit only with exponential decay rates. As a technical step of independent interest, we show an upper bound to the decay of the Belavkin-Staszewski relative entropy upon the application of a conditional expectation.
   
\end{abstract}

\newpage

\tableofcontents


\section{Introduction}


It is now widely established that tools and ideas from quantum information theory can give fresh perspectives to studying complex quantum many-body systems. A notable way this happens is the systematic characterization of the correlations among their constituents, which often allows us to narrow down the complexity of the many-body states in specific ways.

One of the main ways of constraining those correlations is the area law, which states that the amount of information that two adjacent regions share is upper bounded by the size of their mutual boundary. This area law has been shown for gapped ground states \cite{Plenio.2005,Hastings.2.2007,Anshu.3.2021} using entanglement entropy, and also for Gibbs states of finite temperatures and other classes of mixed states \cite{Brandao.2015,Firanko.2023,Arad.2023} using the mutual information \cite{Wolf.2008,Kuwahara.2020,Scalet.2021}. Importantly, area laws can be linked with the efficiency of classical algorithms in the form of tensor network methods \cite{Verstraete.2006,Landau.2015,Kuwahara.2020,Guth.2020}. Another related way in which many-body states are typically constrained is by a fast decay of correlations of distant regions. This is often stated in terms of connected correlation functions \cite{Hastings.2.2006,Brandao.2014,Araki.1969,Kliesch.2014,Froehlich.2015}, but also with other quantifiers such as the mutual information \cite{Bluhm.2022} (although they are often equivalent \cite{Capel.2024,Kochanowski.2023,Bluhm.2024}), as well as measures of quantum entanglement \cite{Kuwahara.2022}. This decay means that distant regions behave roughly independently of each other, so complex collective effects (such as long-range order or entanglement) are absent. 

There is a third, arguably more refined, way in which correlations can be constrained in many-body systems: conditional independence \cite{Hayden.2004,Leifer.2008}. This notion aims to quantify how much the (possibly small) correlations between two separate regions $A$ and $C$ are mediated by a separating region $B$ that shields one from the other, such as in Figure \ref{fig:introABC}.
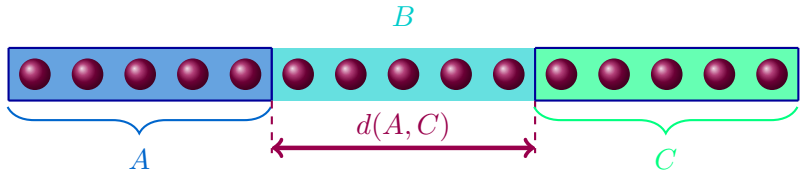
\begin{figure}[ht]
    \begin{center}
        \begin{tikzpicture}[scale=0.7]
            \definecolor{frenchblue}{rgb}{0.0, 0.45, 0.73}
            \Block[5,midblue!60!white,A,1];
            \draw [thick,midblue,decorate,decoration={brace,amplitude=10pt,mirror},xshift=-0.5pt,yshift=-0.6pt](-0.5,-0.6) -- (4.5,-0.6) node[black,midway,yshift=-0.7cm] { \textcolor{midblue}{$A$}};
            \draw [darkblue,thick](-0.5,-0.5) -- (-0.5,0.5) ;
            \draw [darkblue,thick](-0.5,-0.5) -- (4.5,-0.5) ;
            \draw [darkblue,thick](4.5,0.5) -- (-0.5,0.5) ;
            \draw [darkblue,thick](4.5,0.5) -- (4.5,-0.5) ;
            \begin{scope}[xshift=5cm]
                \Block[5,cyan!60!white,B,1];
                \draw [darkblue,thick](-0.5,-0.5) -- (-0.5,0.5) ;
                \draw [dashed,darkred,thick](-0.5,-0.5) -- (-0.5,-1.5) ;
                \draw [dashed,darkred,thick](4.5,-0.5) -- (4.5,-1.5) ;
            \end{scope}
            \begin{scope}[xshift=10cm]
                \Block[5,lightgreen!60!white,C,1];
                \draw [thick,lightgreen,decorate,decoration={brace,amplitude=10pt,mirror},xshift=-0.5pt,yshift=-0.6pt](-0.5,-0.6) -- (4.5,-0.6) node[black,midway,yshift=-0.7cm] { \textcolor{lightgreen}{$C$}};
                \draw [darkblue,thick](-0.5,-0.5) -- (-0.5,0.5) ;
                \draw [darkblue,thick](-0.5,-0.5) -- (4.5,-0.5) ;
                \draw [darkblue,thick](4.5,0.5) -- (-0.5,0.5) ;
                \draw [darkblue,thick](4.5,0.5) -- (4.5,-0.5) ;
            \end{scope}
            \draw [<->,darkred,thick,line width=0.6mm](4.5,-1.4) -- (9.5,-1.4) ;
            \node at (7,1.1) { \textcolor{cyan}{$B$}};
            \node at (7,-1) { \textcolor{darkred}{$d(A,C)$}};
        \end{tikzpicture}
    \end{center}
    \caption{Regions $A$ and $C$ shielded by a region $B$.}
    \label{fig:introABC}
\end{figure}

In quantum systems, the notion of conditional independence of a state $\rho$ is typically expressed through the quantum conditional mutual information (CMI) $I_{\rho}(A:C\vert B)$ \cite{Lieb.1973}, and its behaviour with the geometry of the regions $A,B,C$. This quantity has the following equivalent definitions 
\begin{align} \label{eq:CMI1}
    I_{\rho}(A:C \vert B) &= S(\rho_{AB})+S(\rho_{BC})- S(\rho_{ABC})-S(\rho_{B})
    \\ &= I_{\rho}(A:BC) - I_{\rho}(A: B) \label{eq:CMI2}
     \\ &= H_{\rho}(A | B) - H_{\rho}(A | BC). \label{eq:CMI3}
\end{align}
That is, as a linear combination of subsystem entropies, or as a difference of mutual informations or conditional entropies. 
When $ I_{\rho}(A:C \vert B)=0$, we say that $\rho_{ABC}$ is a quantum Markov state. Additionally, the fact that it is small enables us to conclude that the state $\rho_{ABC}$ can be closely approximated from $\rho_{AB}$ by applying a CPTP map on B alone \cite{Hayden.2004,Fawzi.2015,Sutter.2018}.

The CMI has often appeared in the context of quantum many-body systems. As an analogue of the classical Hammersley-Clifford theorem \cite{Hammersley.1971}, it is known that states in which this quantity vanishes whenever $B$ shields $A$ from $C$ correspond to Gibbs states of commuting Hamiltonians \cite{Leifer.2008,Brown.2012}. Additionally, an exponential decay of the CMI as $B$ increasingly separates $A$ and $C$ has been shown for various instances of Gibbs states \cite{Kato.2019,Kuwahara.2024}, and matrix product states \cite{Svetlichnyy.2022}. A fast decay with the distance between $A$ and $C$ guarantees the accuracy of local reconstruction maps \cite{Fawzi.2015,Junge.2018}, which is linked to the efficiency of quantum algorithms for preparing Gibbs states \cite{Brandao.2019}, and of learning phases of matter \cite{Onorati.2023}. The CMI is also prominently featured in the study of topological order at low energies \cite{Kitaev.2006,Shi.2020}. 

An arbitrarily small CMI is the main way of quantifying conditional independence, but quantum information theory provides us with a framework to systematically construct similar measures, for instance, through generalizations of Eq. \eqref{eq:CMI2} and Eq. \eqref{eq:CMI3}. While their operational meaning might be entirely different \cite{Devetak.2008,Berta.2015,Anshu.2023}, each possible definition comes with the potential for applications. In this paper, we study several such alternative notions of conditional independence. In contrast to the CMI, which is based on the Umegaki relative entropy, ours are largely based on the Belavkin-Staszewksi (BS) relative entropy \cite{Belavkin.1982}. We also study a measure in terms of local Rényi-$2$ entropies inspired by recent results on efficient learning of entropies and entanglement measures \cite{Vermersch.2023}.  

One of our main findings is that the decay of the BS-CMI in Gibbs states is strictly faster than what is believed to hold for any measure of bipartite correlations \cite{Araki.1969,Froehlich.2015,Bluhm.2022}: superexponential, as opposed to just exponential. This difference can be seen in analogy with the classical setting, in which Gibbs distributions have exactly zero CMI, while still having an exponential decay of correlations. In contrast, existing conjectures and results on the quantum CMI only display an exponential decay \cite{Brandao.2019,Kato.2019,Kuwahara.2024}. Our work thus suggests that in quantum Gibbs states there is a separation of the magnitude of conditional independence and decay of correlations, much like in the classical case. 

We then study the application of our notions of conditional independence, and their fast decay, to devise efficient learning schemes for properties of quantum Gibbs states in one dimension. The problem of efficiently learning quantum Gibbs states from local measurements has seen tremendous progress recently, including methods whose efficiency is based on the decay of the CMI \cite{Onorati.2023}. In particular, how to efficiently learn the Hamiltonian with a small number of samples and classical post-processing has been actively studied \cite{Bairey.2019,Anshu.2021,Haah.2022,Rouze.2021,Kokail.2021}, including important experimental efforts \cite{Joshi.2023}, and efficient protocols at all temperatures \cite{Bakshi.2023}. Results along these lines have also appeared for matrix product states in 1D \cite{Cramer.2010,Fanizza.2023}. With our techniques, based on alternative notions of conditional independence, we shift the focus away from Hamiltonian learning and instead study the problem of directly reconstructing tensor network approximations to the state, as well as the global purity.


\subsection{Main results}

In this paper, we focus on Gibbs states of local Hamiltonians in 1D at any inverse positive temperature $\beta>0$. Given a finite chain $\Lambda=ABC$, with $B$ shielding $A$ from $C$ as in Figure \ref{fig:introABC}, and $\rho$ a Gibbs state of a local, translation-invariant Hamiltonian on $\Lambda$, we investigate several notions of conditional independence between $A$ and $C$ conditioned on $B$, and show that they decay with the size of $B$. 

\vspace{0.2cm}

\noindent \textbf{\textit{1. Superexponential decay of the BS Conditional Mutual Information}} 

\vspace{0.1cm}

We first consider generalisations of the CMI involving the Belavkin-Staszewski (BS) relative entropy \cite{Belavkin.1982}, an upper bound for the Umegaki relative entropy, which is given by Eq. \eqref{eq:belavkin-staszewski-relative-entropy}. For three notions of BS Conditional Mutual Information (BS-CMI), termed one-sided (Eq. \eqref{eq:BS-CMI-os}), two-sided (Eq. \eqref{eq:BS-CMI-ts}) and reversed (Eq. \eqref{eq:BS-CMI-rev}), we establish superexponential decay and thereby positively answer a conjecture from \cite{Bluhm.2022}. Note that this is faster than what is expected to hold for the CMI in this context \cite{Kato.2019,Kuwahara.2.2020,Kuwahara.2024}. We prove that all of these quantities exhibit the following scaling:
\begin{equation} \label{eq:main1}
    \widehat I_{\rho}^x(A\mathbin{;}C|B) \le c e^{\alpha|A|}\epsilon(|B|),  \quad x \in \{\operatorname{os}, \operatorname{ts}, \operatorname{rev}\}
\end{equation}
where $A$ and $C$ are separated by $B$. The function $\ell \mapsto \epsilon(\ell)$ decays superexponentially with $\ell$, and $c$ and $\alpha$ are universal for all intervals, only dependent on the inverse temperature and the range and strength of the interaction that gives rise to the local Hamiltonians, as well as the local dimension. The proof of this result is split into two parts, which constitute results of independent interest.

\vspace{0.2cm}

\noindent \textbf{\textit{1.1. Upper bound on the DPI for the BS-entropy}} 

\vspace{0.1cm}

The first step of the proof is an upper bound on the data-processing inequality (DPI) of the BS-entropy for conditional expectations. This bound is expressed in terms of a difference measure that establishes the relationship between $\rho$ and its asymmetric BS-recovery condition \cite{Bluhm.2020}, denoted by $\cB^\sigma_{\cE}(\cdot) := \sigma \cE(\sigma)^{-1}  (\cdot)    $, and analogously with the roles of $\rho$ and $\sigma$ interchanged. More precisely, in \cref{sec:upper-bound-DPI-BS-entropy} we show that 
\begin{equation}\label{eq:DPI1}
    \widehat D(\rho \Vert \sigma) - \widehat D(\cE(\rho) \Vert \cE(\sigma)) \le \begin{cases}
        C(\rho, \sigma, \cE) \, \|\rho \left(\cB^\sigma_{\cE}(\cE(\rho))\right)^{-1} - \1\| \\[1ex]
        C'(\rho, \sigma, \cE) \, \|\sigma \left(\cB^\rho_{\cE}(\cE(\sigma))\right)^{-1} - \1\| 
    \end{cases} \, , 
\end{equation}
with $C(\rho, \sigma, \cE)$ and $C'(\rho, \sigma, \cE)$ multiplicative factors depending on $\rho, \sigma$ and $\cE$. This provides a complementary perspective on the strengthened data-processing inequality established in \cite{Bluhm.2020}. The reason is that it offers an upper limit on the disparity of the BS-DPI, explicitly featuring a difference between the states and their BS-recovery. 

\vspace{0.2cm}

\noindent \textbf{\textit{1.2. Decay of BS-CMI for Gibbs states of 1D local Hamiltonians}} 

\vspace{0.1cm}

To prove Eq.~\eqref{eq:main1}, we combine the bounds of Eq.~\eqref{eq:DPI1} with existing techniques of decay of correlations on Gibbs states of 1D local Hamiltonians, such as those in \cite{Bluhm.2022}. More specifically, in Section \ref{sec:approximate-factorisation}, we include a collection of technical results for these objects based on tight estimates of Araki's expansionals \cite{Araki.1969}, which we use, jointly with \cite[Theorem 5.1]{Bluhm.2022}, to prove:
\begin{equation}\label{eq:approx_fact_BSrecov}
    \norm{\rho_{ABC}\rho^{-1}_{BC} \rho_B \rho_{AB}^{-1} - \1} \leq \varepsilon(|B|) \, ,
\end{equation}
for $\rho$ a Gibbs state on a possibly larger chain $\Lambda'$ than $\Lambda=ABC$, with $\Lambda \subset \Lambda'$, and $\rho_X$ its marginal on $X$ for $X\in \{ AB, B, BC, ABC\}$. Noticing that the BS-recovery of $\rho_{ABC}$ for the partial trace in $A$ is $\rho_{AB} \rho_B^{-1} \rho_{BC} $, we combine Eq. \eqref{eq:approx_fact_BSrecov} with Eq. \eqref{eq:DPI1} to conclude Eq. \eqref{eq:main1}.

\vspace{0.2cm}

\noindent \textbf{\textit{2. Efficient estimation of the global purity of 1D local Hamiltonians}} 

\vspace{0.1cm}

In the next part of the paper, we consider another notion of conditional independence, namely the purity. This quantity can be viewed as a version of the CMI defined in terms of R\'enyi-$2$ entropies. For this purity, using similar techniques and again results from \cite{Bluhm.2022} we resolve a conjecture from \cite{Vermersch.2023} and show that in the same setting as above, it decays exponentially with the size of $|B|$. We subsequently employ this to directly establish the sample and computational efficiency of the scheme to learn the purity from \cite{Vermersch.2023}.

\vspace{0.2cm}

\noindent \textbf{\textit{2.1. Approximate factorisation of the purity}} 

\vspace{0.1cm}

In the first part of this result, contained in Section \ref{sec:purity}, we explicitly prove that
\begin{equation} \label{eq:purity1_intro}
     \left| \frac{\Tr_{AB}[\rho_{AB}^2]\Tr_{BC}[\rho_{BC}^2]}{\Tr_{ABC}[{\rho_{ABC}^2}] \Tr_B[\rho_B^2] } -1\right| \leq c_p e^{- \alpha_p |B|} \, ,
\end{equation}
where $\rho_{ABC}$ is a Gibbs state and the other states are marginals thereof. The constants $c_p$ and $\alpha_p > 0$ are again universal and only depend on the range, inverse temperature, and strength of the interaction, as well as the local dimension.

\vspace{0.2cm}

\noindent \textbf{\textit{2.2. Estimation of the global purity}} 

\vspace{0.1cm}

Subsequently, in Section \ref{sec:MPOpurity}, we consider an $N$-partite system $\Lambda=A_1 \ldots A_N$ and define an $N$-partite version of the purity (removing the Rényi entropy of the total chain), given by 
\begin{equation}\label{eq:Npartite_purity_intro}
    P_2(\rho_{1:N}) = \frac{\prod_{j=1}^{N-1} \Tr_{j:j+1}[\rho_{j:j+1}^2] }{\prod_{j=2}^{N-1} \Tr_{j}[\rho_{j}^2]} \, .
\end{equation}
Then, we show that, due to an iterated application of Eq. \eqref{eq:purity1_intro}, we can obtain an approximation of $\Tr[\rho_{1:N}^2]$ by $P_2(\rho_{1:N})$ up to an error $\varepsilon$ as long as each region $A_i$ has size of order $\mathcal{O}(\log N / \varepsilon)$. Considering next some local measurements of each of the marginals of $\rho$ up to an error $\delta$, we show that there is an algorithm that outputs an estimate of $\Tr[\rho_{1:N}^2]$ up to a multiplicative error $\varepsilon$ with a number of samples and classical post-processing cost of order $\poly (|\Lambda|/\varepsilon)$.
\vspace{0.2cm}

\newpage

\noindent \textbf{\textit{3. Learning of Gibbs states of 1D local Hamiltonians via Matrix Product Operator approximations}} 

\vspace{0.1cm}

Finally, we apply our results for matrix product operator (MPO) reconstructions of Gibbs states and their efficient learning. Again, we consider a multipartite quantum state on an $N$-partite system.
Specifically, the results on the decay of BS-CMI allow us to give an efficient MPO description of 1D Gibbs states.
Furthermore, we show that this description can also be learned efficiently from local tomography.
We find a polynomial sample and computational complexity in both the system size and inverse error.

\vspace{0.2cm}

\noindent \textbf{\textit{3.1. Positive MPO approximations from recovery maps}} 

\vspace{0.1cm}

Apart from the previously mentioned asymmetric recovery map, we introduce an alternative, symmetric, and thereby positive recovery map
\begin{equation}\label{eq:recoveryMapIntro}
    \cR_i(X)=\rho_{i}^{1/2}(\rho_{i}^{-1/2}\rho_{i:{i+1}}\rho_{i}^{-1/2})^{1/2}\rho_{i}^{-1/2}X\rho_{i}^{-1/2}(\rho_{i}^{-1/2}\rho_{i:{i+1}}\rho_{i}^{-1/2})^{1/2}\rho_{i}^{1/2} \, .
\end{equation}
The recovery error is bounded by the BS-CMI with additional terms corresponding to the lowest eigenvalue of marginals of the thermal state and the maximal mutual information \cite{Scalet.2021}, all of which can be appropriately bounded for 1D Gibbs states.
However, this map is not a quantum channel because it is not trace-preserving. This presents a challenge, as the non-contractive nature of the map could lead to exponential amplification of recovery errors from earlier steps in the reconstruction process.
We overcome this issue by proving a Lipschitz constant on a concatenation of maps that is \emph{independent} of the level of concatenation.
We prove a representability result for the Gibbs state by the MPO obtained from concatenating these maps
\begin{equation*}
    \left\|\left(\bigcirc_{i=1}^{N-1}\cR_i\right)(\rho_{1})-\rho_{1:N}\right\|\le\eps
\end{equation*}
with a subpolynomial bond dimension in $|\Lambda|/\eps$.
\vspace{0.2cm}

\noindent \textbf{\textit{3.2. Reconstruction of positive MPO from local tomography}} 

\vspace{0.1cm}
Considering the explicit form of the recovery map Eq.~\eqref{eq:recoveryMapIntro} in terms of the marginals, we show how to learn this MPO representation and prove that the approximation error to the state is robust to small tomographic errors.
We obtain the representation as an explicit formula of the estimated marginals. 
As sublogarithmically-sized marginals are sufficient for the reconstruction, using standard tomography results, the sample and computational cost incurred to reconstruct the state to $\varepsilon$ error in 1-norm is again $\poly (|\Lambda|/\varepsilon)$.
A subpolynomial dependence in system size is also possible using an inherently translation-invariant formulation of the MPO, see \cref{rem:subpolyTomo}.

\vspace{0.2cm}

\noindent \textbf{\textit{Bonus. Exponentially-decaying interactions}}

\vspace{0.1cm}

Many of the techniques used in our main results can or have been recently extended to Gibbs states of Hamiltonians with exponentially-decaying interactions above a threshold temperature, as seen in \cite{Garcia.2023,Capel.2024,Bluhm.2024}. Therefore, a natural question is whether all results presented so far extend to that framework. In \cref{sec:shortrange} we positively answer that question, with small but necessary modifications. The decay of the BS-CMI is now only exponential in $|B|$ and holds only above a critical temperature in line with previous results. In contrast, the estimation of the global purity follows the same arguments as for finite-range interactions. Moreover, we can provide a protocol to learn marginals of Gibbs states in 1D with translation-invariant, exponentially-decaying interactions, by MPO approximations, albeit with a slightly worse bond dimension than for finite-range. Nevertheless, as far as we know, this is the first MPO approximation for Gibbs states in this framework in the literature.

\subsection{Previous work and future directions}
It is interesting to compare our results to existing ones on the representability of thermal states by tensor networks.
In general, it is known that representations with polynomial bond dimension exist for arbitrary lattices \cite{Molnar.2015,Verstraete.2004}.
Furthermore, in one dimension, an MPO description of finite Gibbs states with subpolynomial bond dimension has been shown in \cite{Kuwahara.2020}.
A feature common to all of these constructions is that they give explicit formulas in terms of the Hamiltonian.
While this easily allows for a computation of the MPO given \emph{knowledge of the systems' interactions}, it is unclear how to \emph{learn} such a representation directly. To do that, a precise analysis of the approximation quality given the errors in the estimated Hamiltonian would be needed. Furthermore, even if such a result is achieved, the best existing Hamiltonian learning results \cite{Bakshi.2023,Anshu.2021,Haah.2022} take polynomial time with degrees that are often impractically large.
Defining our reconstruction in terms of measurable marginals, we circumvent these complications and provide a practical formula for the reconstruction.
Another feature of our construction is that it directly works in the thermodynamic limit. Constructions explicitly involving the Hamiltonian would need to be truncated; however, an approach in this direction has been put forward in \cite{Alhambra.2021}.

While a recent result in \cite{Fanizza.2023,Qin.2024}, proposes a way of learning finitely correlated states (which can often be seen as equivalent to matrix product operators) and the application of this framework to the learning of Gibbs states, it is currently unclear whether this approach can yield an equivalent result to ours. The result is conditioned on bounds of singular values of a map involved in the MPO representation of thermal states, which are not known at present.

Finally, on a more conceptual level, the proof technique could be of independent interest in terms of extending the representability result to larger classes of states.
We provide sufficient information theoretic criteria that guarantee efficient MPO approximations. It is a natural question to find out whether other types of states also satisfy these conditions.
The overall idea displays a strong analogy to the technique used in the proof of representability of pure gapped ground states by matrix product states. There, an information-theoretic area law for Rényi entropies \cite{Hastings.2.2007} implies efficient representability of the state by a matrix product state \cite{Verstraete.2006,Schuch.2008} and is also involved in the later rigorous proof of an efficient algorithm for the ground-state energy \cite{Landau.2015}.

Let us finally also comment on an alternative approach that might come to mind when looking at our construction: Using a decay of the standard quantum conditional mutual information and the corresponding recovery channels.
The existence of such recovery channels has been shown \cite{Fawzi.2015} and since they are indeed quantum channels (CPTP-maps) their concatenation should be possible in an analogous way.
There are two obstructions to this route.
Firstly, the channels are not given by an equally simple explicit formula but involve an optimization procedure, which could be cumbersome to deal with on the learning side of our result \cite{Junge.2018,Sutter.2017}.
Secondly, no superexponential decay result is known or believed to hold for the conditional mutual information.

\section{Preliminaries}\label{sec:prelims}

This section is dedicated to introducing and reviewing the basic terminology, notions, and results of quantum systems, entropy measures, Gibbs states in quantum spin chains, and their \emph{approximate factorization}.

\subsection{Basic notation}

First, all vector spaces, tensor products and operator spaces considered in this paper are defined over the field of complex numbers $\C$.
For a quantum system labelled $A$, we let $\cH_A$ denote the Hilbert space. 

A bipartite system $AB$ arising as the composition of two systems $A$ and $B$ is described by the tensor product Hilbert space $\cH_{AB} = \cH_{A}\otimes \cH_{B}$.
In the specific context of quantum spin chains, consecutive letters representing intervals of qubits will generally imply that these intervals are adjacent and non-overlapping.

For $\cH_A$ we use $\cB(\cH_A)$ to denote the space of bounded linear operators with $\norm{\cdot}$ the operator norm. The trace map is $\Tr:\cB(\cH_A) \to \C$, sometimes with subindices if we want to emphasise the system it acts on. A similar notational convention is adapted to denote the identity map $\1$.
In turn, for multipartite systems, we take $\mathrm{tr}_A : \cB(\cH_{AB}) \to \cB(\cH_{B})$ to be the partial trace over $A$.
The spaces of bounded linear operators can be equipped with the Schatten $p$-norms, defined by $\norm{X}_p \coloneq (\operatorname{Tr}(|X|^p))^{1/p}$ for all $p\in [1,\infty)$, where $\abs{X}\coloneq (X^* X)^{1/2}$ stands for the operator square root.
We recall that the Schatten $1$-norm coincides with the usual trace norm, defined as the sum of the singular values of the operator,  and the Schatten $2$-norm is the Hilbert-Schmidt norm arising from the Hilbert-Schmidt inner product on $\cB(\cH_A)$.
Furthermore, in the limiting case $p\to +\infty$ the Schatten $\infty$-norm recovers the operator norm $\norm{\dotarg}_{\infty} = \norm{\dotarg}$, also characterized as the largest singular value of the operator.
We recall that all Schatten $p$-norms are submultiplicative and unitarily invariant. Moreover, they are ordered in the sense that $\norm{X}_{p} \geq \norm{X}_{q}$ for any $1\leq p \leq q \leq \infty$, and satisfy H\"older's inequality $\norm{XY}_r\leq \norm{X}_p \norm{Y}_q$ for all $p,q,r \in [1,\infty]$ with $1/r = 1/p + 1/q$.

Over any system $A$ we consider the non-negative and normalised set $\cS(\cH_{A})\subset \cB(\cH_{A})$ of states or density operators, i.e.~non-negative operators $\rho$ such that $\operatorname{Tr}[\rho] = 1$. In particular, we consider the maximally mixed state $\IdState_{A} \coloneq d_A^{-1}\1_{A}$ on $A$, where $d_A = \dim \cH_A$. We reserve the Greek letters $\rho$ and $\sigma$ for states.
For any multipartite state $\rho_{AB}$, we denote the marginals that arise after application of the partial trace with a subindex that indicates the system they act on, i.e. $\rho_{A}\coloneq \tr_{B}[\rho_{AB}]$ and $\rho_B \coloneq \tr_{A}[\rho_{AB}]$.

As we will be relying on the notion of conditional expectation, we shortly want to introduce it here as well. For a von Neumann subalgebra $\cN$ of $\cB(\cH)$ there exists a unique map, called conditional expectation, which we will generally denote with $\cE:\cB(\cH) \to \cN$, and which satisfies the properties that it projects onto $\cN$ orthogonally w.r.t the Hilbert-Schmidt inner product. An important example and the use case of such maps in this paper are partial traces followed by an embedding, i.e maps of the form $\pi_A \otimes \tr_A: \cB(\cH_{AB}) \to \1_A \otimes \cB(\cH_B) \subseteq \cB(\cH_{AB})$, where $X \in \cB(\cH_{AB})$ is mapped to $\pi_A \otimes \tr_A[X_{AB}]$.

Conditional expectations form a subset within the broader set of quantum channels, encompassing completely positive trace-preserving linear mappings. Note at last that through Stinespring's dilation theorem, every quantum channel can be represented by a composition of an isometry with a partial trace.

\subsection{Entropy measures}\label{sec:entropies}

We introduce several entropy measures that will appear throughout the text, and present some connections between them. Arguably, the most recognized measure among them is the \emph{von Neumann entropy} of a state $\rho$, defined by the expression
\begin{equation}
     S(\rho) \coloneq -\Tr[\rho \log(\rho)] \, .
\end{equation}
Here and in the following we adopt the convention $0 \, \log 0 \equiv 0$.
The von Neumann entropy captures the entropy of a quantum state and is a quantum analogous to the classical Shannon entropy \cite{Shannon.1948}. Also inspired by the classical setting, one can generalise the Kullback-Leibler divergence \cite{Kullback.1951} to the quantum setting. The most prominent of those generalizations is the \emph{Umegaki relative entropy} \cite{Umegaki.1962}. For two quantum states $\rho$ and $\sigma$, it is defined by the expression
\begin{equation}\label{eq:umegakirelative-entropy}
    D(\rho \Vert \sigma) \coloneq 
    \begin{cases}
        \Tr[\rho \log\rho - \rho \log \sigma]
            & \text{ if } \ker \sigma \subseteq \ker \rho \, ,\\
        + \infty
            & \text{ otherwise} \, .
    \end{cases}
\end{equation}
A less prominent one, which has however gained attention in the last years as tool to estimate channel capacities \cite{Fang.2019} and the decay of correlation measures for Gibbs state in 1D \cite{Bluhm.2022} is the \textit{Belavkin-Staszewski relative entropy} \cite{Belavkin.1982} or BS-entropy for short. For any two states $\rho$ and $\sigma$, it is defined as
\begin{equation}\label{eq:belavkin-staszewski-relative-entropy}
    \DBS{\rho}{\sigma} := 
    \begin{cases}
        \Tr[\rho\log (\rho^{1/2} \sigma^{-1} \rho^{1/2})]
            & \text{ if } \ker \sigma \subseteq \ker \rho \, , \\
        + \infty
            & \text{ otherwise} \, .
    \end{cases}
\end{equation}
The reason for terming those entropic measures generalisations is that if the involved states commute the quantum measure reduces to their classical analogue. For example, in the case of the Umegaki and BS-entropy, one would recover the Kullback-Leibler divergence. Hence, both Umegaki and BS-entropy agree in the commuting case, while the BS-entropy is strictly bigger than the Umegaki relative entropy in the case of non-commuting states \cite{Hiai.2017}. Note further that both divergences can be defined for general positive semidefinite matrices, and we will use those definitions under slight abuse of notation.

Whilst not featured in our decay results, we will make use of the maximal Rényi divergence in our results on MPO constructions \cite{Datta.2009}
\begin{equation*}
    D_\infty(\rho\|\sigma)=\log\inf\left\{\lambda:\rho\le\lambda\sigma\right\} \, ,
\end{equation*}
which arises as a limit of a wider family of sandwiched Rényi entropies \cite{Matsumoto.2018}.

As we are mostly concerned with the analysis of quantities involving the BS\nobreakdash-en\-tro\-py, we skip the definitions of the analogues for the relative, respectively von Neumann entropy, and only note that conditional entropy, mutual information, and conditional mutual information were initially defined in terms of the Shannon entropy and then rewritten using the Umegaki relative entropy. Inspired by this latter representation, the BS-entropy analogues are obtained by replacing the Umegaki with the BS-entropy. We refer to \cite{Bluhm.2023, Bluhm.2.2023, Scalet.2021, Zhai.2022} for partial or complete definitions of these quantities.

For any bipartite state $\rho_{AB}$, we define the \emph{BS-conditional entropy} by
\begin{equation}\label{eq:BS-conditional-entropy}
    \HBS{\rho}{A}{B} \coloneq - \DBS{ \rho_{AB} }{ \1_A \otimes \rho_B } , 
\end{equation}
and the \emph{BS-mutual information} by 
\begin{equation}\label{eq:BS-mutual-information}
    \IBS{\rho}{A}{B} \coloneq \DBS{ \rho_{AB} }{\rho_A \otimes \rho_B } . 
\end{equation}
Equivalently, we also introduce the \emph{maximal mutual information} \cite{Scalet.2021}
\begin{equation}\label{eq:max-mutual-information}
    I_{\infty}(A:B) \coloneq D_\infty(\rho_{AB}\|\rho_A \otimes \rho_B ) . 
\end{equation}

The definition of the BS-conditional mutual information (BS-CMI) is relatively more ambiguous and hence we recall the following definitions from \cite{Bluhm.2023}. Let us consider a tripartite quantum system $ABC$. Then, for a quantum state $\rho_{ABC}$ we define the \emph{one-sided BS-conditional mutual information} between the systems $A$ and $C$ conditioned on the system $B$ by the expression
\begin{equation}\label{eq:BS-CMI-os}
    \IosBS{\rho}{A}{C}{B} \coloneq
    \DBS{ \rho_{ABC} }{ \IdState_A \otimes \rho_{BC} } - \DBS{\rho_{AB} }{\IdState_A \otimes \rho_B }, 
\end{equation}
and the \emph{two-sided BS-conditional mutual information} as 
\begin{equation}
    \label{eq:BS-CMI-ts}
    \ItsBS{\rho}{A}{C}{B} \coloneq
    \DBS{\rho_{ABC} }{ \rho_A \otimes \rho_{BC} } - \DBS{\rho_{AB} }{ \rho_A \otimes \rho_B}. 
\end{equation}
Lastly, we define the \textit{reversed BS-conditional mutual information} by
\begin{equation}
    \label{eq:BS-CMI-rev}
    \IrevBS{\rho}{A}{C}{B} := \DBS{\pi_A\otimes \rho_{BC}}{\rho_{ABC}} - \DBS{\pi_A \otimes \rho_B }{\rho_{AB}},
\end{equation}
where we recall that $\pi_A$ stands for the maximally mixed state on $A$.

\subsection{Spin chains in 1D, local Hamiltonians and Gibbs states}

We write $\Lambda \Subset \Z$ for $\Lambda$ being a finite subset of $\Z$, use $|\Lambda|$ to denote its cardinality and $\diam(\Lambda) := \max\{x - y : x, y \in \Lambda\}$ its diameter. With every site $x \in \Z$ we associate a finite-dimensional Hilbert space $\cH_x := \C^d$ with corresponding linear operators $\cA_x := \cB(\cH_x)$. We generalise this concept to finite sets, where for $\Lambda \Subset \Z$, we define the Hilbert space $\cH_\Lambda := \bigotimes_{x \in \Lambda} \cH_x$ and the algebra of linear operators as $\cA_\Lambda := \cB(\cH_\Lambda) = \bigotimes_{x \in \Lambda} \cB(\cH_x)$. The last equality holds due to the finite-dimensional nature of $\cH_\Lambda$. For $\Lambda' \subseteq \Lambda \Subset \Z$, we can therefore consider $X \in \cA_{\Lambda'}$ as an element of $\cA_{\Lambda}$ by identifying $X$ with $X \otimes \1_{\Lambda\backslash\Lambda'}$. This identification is kept implicit and allows us to define the \textit{algebra of local observables} for general $\Sigma \subseteq \Z$ simply by
\begin{equation}
    \cA_\Sigma := \overline{\bigcup\limits_{\Lambda \Subset \Sigma} \cA_\Lambda}^{\norm{\cdot}} \, ,
\end{equation}
where the closure is taken with respect to the operator norm. Now an \textit{interaction} on $\Sigma \subseteq \Z$ is defined as
\begin{equation}
    \Phi:\{\Lambda \Subset \Sigma\} \to \cA_{\Z}, \quad \Lambda \mapsto \Phi(\Lambda) \in \cA_\Lambda \quad \text{with} \quad \Phi(\Lambda) = \Phi(\Lambda)^* \, , 
\end{equation}
and its corresponding local Hamiltonian on any $\Lambda \Subset \Sigma$ as
\begin{equation}
    H_\Lambda := \sum\limits_{\Lambda' \subseteq \Lambda} \Phi(\Lambda') \, . 
\end{equation}
In this paper, we focus on \textit{finite-range interactions}, characterised by the existence of parameters $R > 0$ and $J > 0$ such that $\Phi(\Lambda) = 0$ whenever the diameter of $\Lambda$ exceeds $R$, and 
\begin{equation*}
    \sum\limits_{\Lambda \Subset \Z \;:\; x \in \Lambda} \norm{\Phi(\Lambda)} \leq J  \, , 
\end{equation*}
for all $x \in \Z$. In the literature, it is common to alternatively impose the condition that for all $\Lambda \Subset \Z$, $\norm{\Phi(\Lambda)} \leq \widetilde{J}$, which can be related to $J$ by observing that, due to the interaction vanishing if the diameter of the input exceeds $R$, 
\begin{equation*}
    \sum\limits_{\Lambda \Subset \Z \;:\; x \in \Lambda} \norm{\Phi(\Lambda)} = \sum\limits_{k = 0}^{R}\sum\limits_{\Lambda \in \cP(\{x - k, \hdots, x - k + R\}\backslash\{x\})} \norm{\Phi(\{x\} \cup \Lambda)} \leq (R + 1)2^R \widetilde{J} \, . 
\end{equation*}
In addition to the finite range, we further require interactions to exhibit \textit{translation invariance}. In \Cref{sec:shortrange}, we extend all our results to the setting of exponentially-decaying interactions, whose formalism we introduce therein.

Lastly, we define the \textit{Gibbs state} for a Hamiltonian $H \in \cB(\cH)$ with $H = H^*$ over an arbitrary (finite-dimensional) Hilbert space $\cH$ at inverse temperature $\beta \in (0, \infty)$ as
\begin{equation}
    \rho_{\cH}^\beta[H] := \frac{e^{-\beta H}}{\Tr[e^{-\beta H}]} \, . 
\end{equation}
In \cref{sec:superexponential-decay-BS-CMI}, i.e., when we prove superexponential decay of Gibbs states of local translation-invariant Hamiltonians on a quantum spin system, we drop the index of the Hilbert space, absorb the inverse temperature into the Hamiltonian and write for $\Lambda \Subset \Sigma \subseteq \Z$ just
\begin{equation}
    \rho^{\Lambda} := \rho^\beta_{\cH_\Lambda}[H_\Lambda] \, . 
\end{equation}
Notice that Gibbs states of finite $\beta$ are always full rank.

\subsection{Approximate factorisation of Gibbs states of local Hamiltonians in 1D}\label{sec:approximate-factorisation}

We now present the technical lemmas we need for the decay of correlations and uniformity of 1D Gibbs states. While they hold at all temperatures for any finite-range interaction, the constants involved depend on range, interaction strength and temperature.
This dependence is mostly the same for all involved quantities, so we introduce the following convenient notation.

Let $\Lambda \Subset \mathbb{Z}$ be a finite interval. Let us split $\Lambda$ into two subintervals $X$ and $Y$ so that $\Lambda = X \cup Y$ or $\Lambda = XY$ for short and write
\begin{equation}\label{def:expantionals}
     E_{X,Y} \, := \, e^{-\, H_{XY}} \, e^{\, H_{X} \, + \, H_{Y}} \,. 
\end{equation}
Note that $ E^*_{X,Y}= e^{\, H_{X} \, + \, H_{Y}}  \,  e^{-\, H_{XY}} $ and  $ E^{-1}_{X,Y}= e^{- H_{X} \, - \, H_{Y}}  \,  e^{\, H_{XY}} $ and that we absorbed $\beta$ into the Hamiltonian for better readability.  The following proposition is extracted from \cite{Bluhm.2022,Garcia.2023} and contains an alternative formulation of Araki's results for estimates on expansionals \cite{Araki.1969}.

\begin{proposition}[{\cite[Corollary 3.4]{Bluhm.2022}}]\label{prop:estimates_expansionals}
    Let $\Phi$ be an interaction of finite-range $R$ and strength $J$ over $\Z$, at any inverse temperature $\beta < \infty$, which is further translation invariant. Then the following hold:
    \begin{enumerate}
        \item[(i)] There is an absolute constant $\mathcal{G}>1$ depending only on $J$, $R$ and $\beta$ such that, for any finite interval $\Lambda = X Y \Subset \mathbb{Z}$ split into two subintervals $X$ and $Y$, we have:
        \begin{equation*}
            \norm{E_{X,Y}} \, , \,\norm{E_{X,Y}^{-1}} \,\, \leq \,\, \cG \, .
        \end{equation*}
        \item[(ii)] There is a positive and decreasing function $\ell \mapsto \delta(\ell)$ with superexponential decay and depending on $J$, $R$ and $\beta$ such that if we add two intervals $\widetilde{X}$ and $\widetilde{Y}$ adjacent to $X$ and $Y$, respectively, so that we get a larger interval $\widetilde{X}XY\widetilde{Y}$, then
        \begin{equation*}
            \left\| E_{X,Y}^{-1} - E^{-1}_{\widetilde{X}X,Y\widetilde{Y}} \right\|, \left\| E_{X,Y} - E_{\widetilde{X}X,Y\widetilde{Y}} \right\| \, \leq \, \delta (\ell) \,.
        \end{equation*}
        for any $\ell \in \mathbb{N}$ such that $\ell \, \leq \, |X| \, , \, |Y|$. 
     \end{enumerate}
\end{proposition}

\noindent Note in addition that for $\Lambda' \subset \Lambda \Subset \mathbb{Z}$ with local Hamiltonian $H_{\Lambda'}$ supported in $\cH_{\Lambda'}$ but lifted to $\cH_{\Lambda}$, the map
\begin{equation}
    \cA_{\Lambda} \to \cA_{\Lambda}\,, \quad Q \mapsto \tr_{\Lambda'}[e^{-H_{\Lambda'}} Q] = \tr_{\Lambda'}[e^{-\frac{1}{2}H_{\Lambda'}} Q e^{-\frac{1}{2}H_{\Lambda'}}]
\end{equation}
where $\tr_{\Lambda'}$ denotes the partial trace in $\Lambda'$, has the following property:
\begin{equation}
    \norm{\tr_{\Lambda'}[e^{-H_{\Lambda'}} Q]} \le \norm{\tr_{\Lambda'}[e^{-H_{\Lambda'}}]} \, \norm{Q} = \Tr_{\Lambda'}[e^{-H_{\Lambda'}}] \norm{Q} \,.
\end{equation}
As a consequence,
\begin{equation}\label{eq:contractiveExpectation}
    \cA_{\Lambda} \to \cA_{\Lambda}\,, \quad Q \mapsto \tr_{\Lambda'}[\rho^{\Lambda'} Q] 
\end{equation}
is positive, unital and hence contractive. This observation is essential for the following result, in which we are considering an interval $\Lambda$ split into three subintervals $\Lambda=ABC$, where $B$ shields $A$ from $C$, as illustrated in Figure \ref{fig:ABC}.

\begin{figure}[ht]
    \begin{center}
        \begin{tikzpicture}[scale=0.7]
            \definecolor{frenchblue}{rgb}{0.0, 0.45, 0.73}
            \Block[5,midblue!60!white,A,1];
            \draw [thick,midblue,decorate,decoration={brace,amplitude=10pt,mirror},xshift=-0.5pt,yshift=-0.6pt](-0.5,-0.6) -- (4.5,-0.6) node[black,midway,yshift=-0.7cm] { \textcolor{midblue}{$A$}};
            \begin{scope}[xshift=5cm]
                \Block[5,cyan!60!white,B,1];
                \draw [thick,cyan,decorate,decoration={brace,amplitude=10pt,mirror},xshift=-0.5pt,yshift=-0.6pt](-0.5,-0.6) -- (4.5,-0.6) node[black,midway,yshift=-0.7cm] { \textcolor{cyan}{$B$}};
            \end{scope}
            \begin{scope}[xshift=10cm]
                \Block[5,lightgreen!60!white,C,1];
                \draw [thick,lightgreen,decorate,decoration={brace,amplitude=10pt,mirror},xshift=-0.5pt,yshift=-0.6pt](-0.5,-0.6) -- (4.5,-0.6) node[black,midway,yshift=-0.7cm] { \textcolor{lightgreen}{$C$}};
            \end{scope}
            \node at (15,0.8) {\huge $\Lambda$};
        \end{tikzpicture}
      \end{center}
      \caption{Representation of an interval $\Lambda$ split into three subintervals $\Lambda=ABC$, where $B$ shields $A$ from $C$.}
      \label{fig:ABC}
\end{figure}
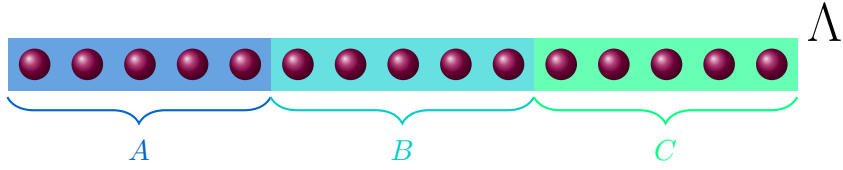

\begin{proposition}[{\cite[Section 3.4]{Bluhm.2022}}]\label{prop:BoundingPartialTraceInverses}
    In the context of \cref{prop:estimates_expansionals}, let $\Lambda = ABC \Subset \Z$ (we admit the possibility of some subintervals being empty). Then, there is an absolute constant $\mathcal{C}$ depending only on the strength $J$, range $R$ of the interaction and inverse temperature $\beta$, such that
    \begin{align}
       &  \big\| \tr_{B}[\rho^{B}Q] \big\|, \big\| \tr_{B}[\rho^{B}Q]^{-1} \big\| \, \leq \, \mathcal{C} \, \,  , \quad \quad Q \in \{ E_{B,C}^{*} \, , \, E_{B,C}\, , \, E_{A,B}^{*} \, , \, E_{A,B} \} \, , \label{equa:BoundingPartialTraceInverses1}\\[2mm]
      &   \label{equa:BoundingPartialTraceInverses2} \big\| \tr_{AB}[\rho^{AB}Q] \big\| \, , \, \big\| \tr_{AB}[\rho^{AB}Q]^{-1} \big\| \,   \leq \, \mathcal{C} \, \,  , \quad \quad Q \in \{ E_{A,B}^{* \, -1} \, , \, E_{A,B}^{\, -1} \} \, ,\\[2mm]
      &  \label{equa:BoundingPartialTraceInverses3} \big\| \tr_{B}\big[\rho^{B}E_{A,B}^{*} E_{AB,C}^{*}\big] \big\| \, , \, \big\| \tr_{B}\big[\rho^{B}E_{A,B}^{*} E_{AB,C}^{*}\big]^{-1} \big\| \, \leq \, \mathcal{C} \, .
    \end{align}
\end{proposition}

These two propositions are instrumental in the proof of the following result, which provides an approximate factorization of the Gibbs state of a finite-range, translation-invariant Hamiltonian. 

\begin{theorem}[{\cite[Eq. (17) in Theorem 5.1]{Bluhm.2022}}]\label{lemma:approximate-factorization-of-Gibbs-state}
    For $\Phi$ a finite-range, translation-invariant interaction over $\Z$ there exists a positive function $\ell \mapsto \varepsilon(\ell)$ with superexponential decay in $\ell$, depending only on $J$, $R$, and $\beta$ such that for every $\Lambda \Subset \Z$ split into three subintervals $\Lambda = ABC$, where $B$ shields $A$ from $C$ and $|B| \ge \ell$ for its local Gibbs state $\rho^{\Lambda} =: \rho_{ABC}$ it holds that
    \begin{equation}\label{eq:approximate-factorization-of-Gibbs-state}
        \norm{\rho_{ABC} \rho_{BC}^{-1} \rho_B \rho_{AB}^{-1} - \1} \le \varepsilon(\ell) \, .
    \end{equation}
    The form of the superexponentially-decaying function is given by 
    \[
    \varepsilon(\ell)=\cC_1\frac{\cC_2^{1+\lfloor\ell/2\rfloor}}{(\lfloor\lfloor\ell/2\rfloor/R\rfloor+1)!}
    \]
    for constants $\cC_1,\cC_2$ that only depend on $J$, $R$, and $\beta$.
\end{theorem}
Note that, in Eq.~\eqref{eq:approximate-factorization-of-Gibbs-state},  only $\rho_{ABC}$ denotes a Gibbs state of a local Hamiltonian while $\rho_{AB}, \rho_{BC}$ and $\rho_B$ are marginals after partially tracing out systems. We further suppressed tensoring with identity, i.e., $\rho_{AB}$ for example has to be understood as $\rho_{AB} \otimes \1_C$.

In the rest of the paper, we will need a refined version of this theorem for Gibbs states on $\Lambda=A'ABCC'$, for which we only compare the marginals in $AB$, $B$, $BC$ and $ABC$ as above. To generalize this, however, we need to introduce the following lemmata. The first constitutes a generalization of Eq.~\eqref{equa:BoundingPartialTraceInverses1} in Proposition \ref{prop:BoundingPartialTraceInverses}.
 
\begin{figure}[h]
    \begin{center}
        \begin{tikzpicture}[scale=0.68]
            \definecolor{frenchblue}{rgb}{0.0, 0.45, 0.73}
            \Block[5,midblue!60!white,A,1];
            \draw [thick,midblue,decorate,decoration={brace,amplitude=10pt,mirror},xshift=-0.5pt,yshift=-0.6pt](-0.5,-0.6) -- (4.5,-0.6) node[black,midway,yshift=-0.7cm] { \textcolor{midblue}{$A$}};
            \begin{scope}[xshift=5cm]
                \Block[5,cyan!60!white,B,1];
                \draw [thick,cyan,decorate,decoration={brace,amplitude=10pt,mirror},xshift=-0.5pt,yshift=-0.6pt](-0.5,-0.6) -- (4.5,-0.6) node[black,midway,yshift=-0.7cm] { \textcolor{cyan}{$B$}};
            \end{scope}
            \begin{scope}[xshift=10cm]
                \Block[5,lightgreen!60!white,C,1];
                \draw [thick,lightgreen,decorate,decoration={brace,amplitude=10pt,mirror},xshift=-0.5pt,yshift=-0.6pt](-0.5,-0.6) -- (4.5,-0.6) node[black,midway,yshift=-0.7cm] { \textcolor{lightgreen}{$C$}};
            \end{scope}
            \begin{scope}[xshift=15cm]
                \Block[4,midgreen!60!white,D,1];
                \draw [thick,midgreen,decorate,decoration={brace,amplitude=10pt,mirror},xshift=-0.5pt,yshift=-0.6pt](-0.5,-0.6) -- (3.5,-0.6) node[black,midway,yshift=-0.7cm] { \textcolor{midgreen}{$C'$}};
            \end{scope}
            \begin{scope}[xshift=-4cm]
                \Block[4,darkblue!60!white,B,1];
                \draw [thick,darkblue,decorate,decoration={brace,amplitude=10pt,mirror},xshift=-0.5pt,yshift=-0.6pt](-0.5,-0.6) -- (3.5,-0.6) node[black,midway,yshift=-0.7cm] { \textcolor{darkblue}{$A'$}};
            \end{scope}
            \node at (19,0.8) {\huge $\Lambda$};
        \end{tikzpicture}
    \end{center}
    \caption{Representation of an interval $\Lambda$ split into five subintervals $\Lambda=A'ABCC'$.}
    \label{fig:AABCC}
\end{figure}
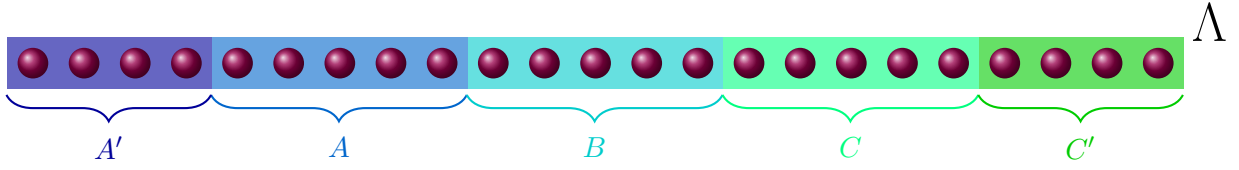

\begin{lemma}\label{Lemma:Modified_BoundingPartialTraceInverses}
Let $\Lambda \Subset \mathbb{Z}$ be a finite interval split into five subintervals $\Lambda=A'ABCC'$ as in Figure \ref{fig:AABCC}. Then, for $\rho = \rho^{\Lambda}$ the Gibbs state on $\Lambda$, and for every $\beta>0$, the following holds:
\begin{equation}\label{eq:aux_inverse_norm_contraction}
    \norm{\tr_{AC}[\rho^{A}\otimes \rho^{C} Q]^{-1}} \,   \leq \, \mathcal{C} \, 
\end{equation}
for $Q= E_{AB,C} \, E_{A,B}$ and a constant $\cC$ that only depends on $J$, $R$, $\beta$ and $d$.
\end{lemma}

We defer the proof of \Cref{Lemma:Modified_BoundingPartialTraceInverses} to \Cref{sec:proofsApproxTens}. Next, we show that the decay provided by the measure of conditional independence in \Cref{lemma:approximate-factorization-of-Gibbs-state} can be reduced from the total Gibbs state to a large marginal of it. The proof of this lemma is also deferred to \Cref{sec:proofsApproxTens}.

\begin{lemma}\label{Lemma:DPI_mixing_condition}
Let $\Lambda \Subset \mathbb{Z}$ be a finite interval split into five subintervals $\Lambda=A'ABCC'$ as in Figure \ref{fig:AABCC}. Then, for $\rho = \rho^{\Lambda}$ the Gibbs state on $\Lambda$, and for every $\beta>0$, the following holds
\begin{equation}\label{eq:DPI_mixing_condition}
   \left\| \rho_{AB} \, \rho_{B}^{-1}  \rho_{BC} \, \rho_{ABC}^{-1} - \1 \right\| < \mathcal{C}^4  \, \left\| \rho_{A'AB} \, \rho_{B}^{-1} \rho_{BCC'} \, \rho_{A'ABCC'}^{-1} - \1 \right\| \, ,  
\end{equation}
for $\mathcal{C}>1$ depending on $R, J$ and $\beta$, and given in Proposition \ref{prop:BoundingPartialTraceInverses}.

\end{lemma}

\begin{remark}
Note that Eq. \eqref{eq:DPI_mixing_condition} can be interpreted as a modified data-processing inequality for the partial trace for the mixing condition between two positive states $\eta$ and $\sigma$,
\begin{equation*}
    \norm{\eta \, \sigma^{-1} - \1} \, ,
\end{equation*}
with contraction coefficient upper bounded by $\mathcal{C}^4$, identifying $\eta:= \rho_{A'AB} \, \rho_{B}^{-1} \rho_{BCC'} $, $\sigma := \rho_{A'ABCC'}$, and taking the partial trace in $A'C'$.
\end{remark}

Combining now the findings of Theorem \ref{lemma:approximate-factorization-of-Gibbs-state} and Lemma \ref{Lemma:DPI_mixing_condition}, we can conclude the following:

\begin{corollary}\label{cor:Clean_approx_identity_prodmarginals}
For $\Phi$ a finite-range, translation-invariant interaction over $\Z$, there exists a positive function $\ell \mapsto \tilde \varepsilon(\ell)$ with superexponential decay in $\ell$, depending only on $J$, $R$ and $\beta$, such that for every $\Lambda \Subset \Z$ split into three subintervals $\Lambda = A'ABCC'$, with $|B| \ge \ell$, for its local Gibbs state $\rho^{\Lambda} =: \rho_{A'ABCC'}$ it holds that
    \begin{equation}\label{eq:clean_approx_identity_prodmarginals}
        \norm{\rho_{ABC} \rho_{BC}^{-1} \rho_B \rho_{AB}^{-1} - \1} \le \tilde \varepsilon(\ell) \, .
    \end{equation}
\end{corollary}

In addition to the above, the authors in \cite[Theorem 6.2]{Bluhm.2022} showed exponential decay of correlations between spatially separated regions. This will become relevant to us later in the form of local indistinguishability, as per \cref{theorem:local_indistinguishability}.

\begin{proposition}[Decay of correlations {\cite[Theorem 6.2]{Bluhm.2022}}]\label{prop:decay_correlations}
    For $\Phi$ a finite-range, translation-invariant interaction over $\Z$ there exist  positive constants $c$, $\alpha$, depending only on $J$, $R$ and $\beta$ such that for every $\Lambda \Subset \Z$ split into three subintervals $\Lambda = ABC$, where $B$ shields $A$ from $C$ and for $O_A \in \cA_A$ and $O_C \in \cA_C$  it holds that
    \begin{align*}
        \left|\Tr_{ABC}[\rho^{ABC} O_A \, O_C ] -  \Tr_{ABC}[\rho^{ABC} O_A  ]  \Tr_{ABC} [\rho^{ABC} O_C  ]  \right| & \leq \norm{O_A} \norm{O_C} \, c \, e^{- \alpha |B|} \, .
    \end{align*}
\end{proposition}

The exponential decay of correlations can be used to prove local indistinguishability of Gibbs states, a result that can be found in \cite[Proposition 7.1]{Bluhm.2022}. 

\begin{theorem}[Local indistinguishability {\cite[Proposition 7.1]{Bluhm.2022}}]\label{theorem:local_indistinguishability}
    For $\Phi$ a finite-range, translation-invariant interaction over $\Z$, there exist  positive constants $c'$, $\alpha'$, depending only on $J$, $R$ and $\beta$, such that for every $\Lambda \Subset \Z$ split into three subintervals $\Lambda = ABC$, where $B$ shields $A$ from $C$, and for $O_A \in \cA_A$ and $O_C \in \cA_C$  it holds that
    \begin{align*}
        \left| \Tr_{ABC}[\rho^{ABC} O_A] -  \Tr_{AB}[\rho^{AB} O_A]  \right| & \leq \norm{O_A} \, c' \, e^{-\alpha' |B|} \, , \\[1mm]
         \left| \Tr_{ABC}[\rho^{ABC} O_C] -  \Tr_{BC}[\rho^{BC} O_C]  \right| & \leq \norm{O_C} \, c' \, e^{-\alpha' |B|} \, .
    \end{align*}
\end{theorem}

All these results will be essential for our proof that the purity of a Gibbs state of a translation-invariant, local Hamiltonian in 1D decays exponentially fast with the size of the middle system (cf. \Cref{sec:purity}). 

\begin{remark}
    It is possible to extend all the results to the context of short-range (i.e. exponentially-decaying) interactions, as shown in \cite{Garcia.2023,Capel.2024}. We review all these extensions and provide necessary new ones in \cref{sec:shortrange}.
\end{remark}

\section{Main results}

%
\subsection{An upper bound on the DPI for the Belavkin Staszewski entropy}\label{sec:upper-bound-DPI-BS-entropy}
The search for the presence and tightness of data-processing inequality (DPI) in entropy measures, namely that they cannot increase under quantum channels, constitute relevant problems in quantum Shannon theory. For the Umegaki relative entropy, it was originally shown by Petz \cite{Petz.1988,Petz.2003} that saturation of the DPI is equivalent to recoverability of one state in terms of the other, through the so-called \textit{Petz recovery map}. A strengthening of the DPI for the relative entropy, namely lower bounds to the DPI after applying a conditional expectation (or in general, a quantum channel) have been obtained in the past few years \cite{Sutter.2017,Carlen.2020,Gao.2021}, starting with the breakthrough of \cite{Fawzi.2015}. These results have also been extended to the larger families of standard f-divergences \cite{Carlen.2018,Hiai.2017}. Results in the converse direction, i.e. upper bounds for the DPI, can be found e.g. in \cite{Bluhm.2023}. 

We can now ask a similar question for the BS-entropy. This quantity also satisfies a DPI, which was proven to be saturated if and only if \cite{Bluhm.2020}
\begin{equation}\label{eq:asymmetric_BS-recovery}
     \widehat D(\rho \Vert \sigma) - \widehat D(\cE(\rho) \Vert \cE(\sigma)) = 0 \Leftrightarrow  \sigma =  \rho \, \mathcal{E}(\rho)^{-1} \cE(\sigma) = \cB^\rho_{\cE}(\mathcal{E}(\sigma)) \, . 
\end{equation}
The term on the right-hand side above was coined \textit{BS-recovery condition}. In the same paper, a strengthened DPI for the BS-entropy (and any maximal f-divergence) in terms of the distance between a state and its BS-recovery was provided. 

Here, we would like to complement this picture by giving an upper bound on the DPI of the BS-entropy. More specifically, we provide an upper bound for $\widehat D(X \Vert Y) - \widehat D(\cE(X) \Vert \cE(Y))$ in terms of distance measures that explicitly feature the recovery condition.
We use the usual representation of the logarithm we want to recall here. For $X > 0$, we have that 
\begin{equation}\label{eq:integral-representation-logarithm}
    \log(X) = \int_{0}^\infty \left( \frac{1}{t + 1} - \frac{1}{t + X} \right) \, dt \, . 
\end{equation}
Using this representation we derive the following simple but essential lemma. 
\begin{lemma}\label{lemma:difference-estimate}
     Let $X, Y \in \cB(\cH)$ with $X > 0$, $Y > 0$ and $\cE: \cB(\cH) \to \cB(\cH)$ a conditional expectation. We further set $Z = X^{1/2} Y^{-1} X^{1/2}$, $X_{\cE} := \cE(X)$, $Y_{\cE} := \cE(Y)$ and finally $Z_{\cE} = X_{\cE}^{1/2} Y_{\cE}^{-1} X_{\cE}^{1/2}$. For $T > 0$, we find that
     \begin{equation*}
         \int_{0}^T \Tr[X_{\cE} \frac{1}{t + Z_{\cE}} - X \frac{1}{t + Z}] dt = \int_{0}^T \Tr[X_{\cE}^{1/2} \frac{1}{t + Z_{\cE}} X_{\cE}^{-1/2}(XY^{-1} - X_{\cE} Y_{\cE}^{-1}) X^{1/2} \frac{1}{t + Z} X^{1/2}] \, dt \, . 
     \end{equation*}
\end{lemma}
\begin{proof}
    We find by cyclicity of the trace and the fact that $\cE$ is a conditional expectation that
    \begin{align*}
        \int_{0}^T \Tr[X_{\cE} \frac{1}{1 + Z_{\cE}} - X \frac{1}{t + Z}] dt &= \int_{0}^T \Tr[ X_{\cE}^{1/2}\frac{1}{t + Z_{\cE}} X_{\cE}^{-1/2} X - X^{1/2} \frac{1}{t + Z} X^{-1/2} X] \, dt \, . \\
        &= \int_{0}^T \Tr[\left(X_{\cE}^{1/2}\frac{1}{t + Z_{\cE}} X_{\cE}^{-1/2} - X^{1/2} \frac{1}{t + Z} X^{-1/2}\right) X] \, dt  \, . 
    \end{align*}
    Recalling that, for any two invertible matrices $A$ and $B$, one can write $A^{-1}-B^{-1}= A^{-1}(B-A)B^{-1}$, we note that 
    \begin{align*}
        X_{\cE}^{1/2}&\frac{1}{t + Z_{\cE}} X_{\cE}^{-1/2} - X^{1/2} \frac{1}{t + Z} X^{-1/2}\\
        &= X_{\cE}^{1/2}\frac{1}{t + Z_{\cE}} X_{\cE}^{-1/2}(X^{1/2}(Z + t) X^{-1/2} - X_{\cE}^{1/2}(Z_{\cE} + t) X_{\cE}^{-1/2}) X^{1/2} \frac{1}{t + Z} X^{-1/2}\\
        &= X_{\cE}^{1/2}\frac{1}{t + Z_{\cE}} X_{\cE}^{-1/2}(XY^{-1} - X_{\cE} Y_{\cE}^{-1}) X^{1/2} \frac{1}{t + Z} X^{-1/2}
    \end{align*}
    and inserting it into the above equation gives the claim.
\end{proof}
With this lemma at hand, we are set to prove the upper bound on the BS-entropy DPI.
\begin{theorem}\label{theorem:upper-bound-DPI-of-BS-entropy}
    Let $X, Y \in \cB(\cH)$ with $X > 0$, $Y > 0$ and $\cE:\cB(\cH) \to \cB(\cH)$ a conditional expectation. Then, we have
    \begin{equation}\label{eq:DPI-BS-entropy-estimate-1}
        \begin{aligned}
            \widehat D(X \Vert Y) - \widehat D(\cE(X) \Vert \cE(Y)) &\le \norm{X^{-1/2} Y X^{-1/2}} \norm{X}_1 \norm{\cE(X)^{1/2}}\norm{\cE(X)^{-1/2}}\\
            &\hphantom{=} \cdot \norm{\cE(Y)^{-1}\cE(X)}\norm{X Y^{-1}\cE(Y)\cE(X)^{-1} - \1} \, . 
        \end{aligned}
    \end{equation}
    and 
    \begin{equation}\label{eq:DPI-BS-entropy-estimate-2}
        \begin{aligned}
            \widehat D(X \Vert Y) - \widehat D(\cE(X) \Vert \cE(Y)) &\le \norm{X^{-1/2} Y X^{-1/2}} \norm{X}_1 \norm{\cE(X)^{1/2}}\norm{\cE(X)^{-1/2}}\\
            &\hphantom{=} \cdot \norm{Y^{-1}X}\norm{YX^{-1} \cE(X) \cE(Y)^{-1}- \1} \, . 
        \end{aligned}
    \end{equation}
\end{theorem}
\begin{proof}
    Both inequalities follow from the same estimate through the application of Hölder inequality. To obtain this expression we employ \cref{eq:integral-representation-logarithm} using the notation of \cref{lemma:difference-estimate}
    \begin{align*}
        \widehat D(X \Vert Y) - \widehat D(X_{\cE} \Vert Y_{\cE}) = \lim_{T \to \infty} \int_{0}^T \Tr[X_{\cE} \frac{1}{t + Z_{\cE}} - X \frac{1}{t + Z}] \, dt \, . 
    \end{align*}
    Further, using \cref{lemma:difference-estimate} and H\"older's inequality we can upper bound
    \begin{align*}
        &\int_{0}^T \Tr[X_{\cE} \frac{1}{t + Z_{\cE}} - X \frac{1}{t + Z}] dt \\
        &\hspace{1.5cm} \le \int_{0}^T \norm{X_{\cE}^{-1/2}} \frac{1}{t + \norm{Z_{\cE}^{-1}}^{-1}} \norm{X_{\cE}^{1/2}} \norm{X Y^{-1} - X_{\cE}Y_{\cE}^{-1}} \norm{X^{1/2}}_2 \frac{1}{t + \norm{Z^{-1}}^{-1}}\norm{X^{1/2}}_2 \, dt \, .
    \end{align*}
    Note also that $\norm{\frac{1}{t + Z}} = \frac{1}{t + \norm{Z^{-1}}^{-1}}$ and analogously for the fraction including $Z_{\cE}$. Utilising DPI for $(A, B) \mapsto \norm{A^{1/2}B^{-1}A^{1/2}}$ (see e.g. \cite[Proposition 4.7]{Tomamichel.2016}), we can estimate $\frac{1}{t + \norm{Z_{\cE}^{-1}}^{-1}} \le \frac{1}{t + \norm{Z^{-1}}^{-1}}$. 
        \begin{align*}
       &  \int_{0}^T \Tr[X_{\cE} \frac{1}{t + Z_{\cE}} - X \frac{1}{t + Z}] dt \\
       &\hspace{2.5cm}\le \norm{X_{\cE}^{-1/2}}\norm{X_{\cE}^{1/2}} \norm{X^{1/2}}_2^2 \norm{X Y^{-1} - X_{\cE}Y_{\cE}^{-1}}   \int_{0}^T   \frac{1}{\left(t + \norm{Z^{-1}}^{-1}\right)^2 }   \, dt \, .
    \end{align*}
    Finally, the fact that $\norm{X^{1/2}}_2^2 = \norm{X}_1$, integrating both sides, and then taking the limit $T \to \infty$ give
    \begin{align}\label{eq:upperBoundDPI}
        &\widehat D(X \Vert Y) - \widehat D(X_{\cE} \Vert Y_{\cE}) \le \norm{X_{\cE}^{-1/2}} \norm{X^{-1/2}Y X^{-1/2}}  \norm{X_{\cE}^{1/2}} \norm{X Y^{-1} - X_{\cE}Y_{\cE}^{-1}} \norm{X}_1. 
    \end{align}
\end{proof}

\begin{remark}
    The above result can readily be extended to quantum channels instead of conditional expectations. The statement and proof of that result can be found in the appendix in \cref{corollary:extending-upper-bound-bs-entropy-to-channels}.
\end{remark}

\begin{remark}\label{rem:complementing-perspective-BS-DPI}
    As mentioned in the introduction, the above result complements the findings of \cite{Bluhm.2020}, providing a comprehensive perspective on the data processing inequality of the BS-entropy. Using the notation of the theorem for quantum states $\rho$ and $\sigma$, and the BS-recovery for $\sigma$ denoted by $\cB^\rho_{\cE}(\cE(\sigma)) = \rho \, \mathcal{E}(\rho)^{-1} \cE(\sigma)$, the following chain of inequalities hold:
    \begin{align*}
        \left(\frac{\pi}{8}\right)^4 \norm{\rho^{-1/2} \sigma \rho^{-1/2}}^{-4} & \norm{\cE(\rho)^{-1} }^{-2} \norm{\cB^{\rho}_{\cE}(\cE(\sigma)) - \sigma}_2^4 \\
        &\le \widehat{D}(\rho \Vert \sigma) - \widehat{D}(\cE(\rho) \Vert \cE(\sigma)) \le \\
        &\hphantom{\le \widehat{D}}\norm{\rho^{-1/2}\sigma\rho^{-1/2}} \norm{\cE(\rho)^{-1} }^{1/2} \norm{\cE(\sigma)^{-1} \cE(\rho)} \norm{\sigma^{-1}}\norm{   \cB_{\cE}^\rho(\cE(\sigma)) - \sigma} \, , 
    \end{align*}
    where the upper bound follows from Equation~\eqref{eq:upperBoundDPI}.
\end{remark}


\subsection{Superexponential decay of the BS-CMI}\label{sec:superexponential-decay-BS-CMI}
In this section, we prove one of our main results: the superexponential decay of the three definitions of the BS-CMI. The proof crucially involves Theorem \ref{theorem:upper-bound-DPI-of-BS-entropy}, as well as the results from Sec. \ref{sec:approximate-factorisation}. Before we do so, we need a technical lemma concerning norm estimates of functions of Gibbs states of local translation invariant Hamiltonians on a spin chain. 

\begin{lemma}\label{lemma:norm-estimates-of-functions-of-Gibbs-states}
     Let $\Phi$ be a finite-range, translation invariant-interaction over $\Z$ and $\Lambda \Subset \Z$ split into three subintervals $\Lambda = ABC$, where $B$ shields $A$ from $C$, with local Gibbs state $\rho^{\Lambda} = \rho_{ABC}$. Then
     \begin{enumerate}
        \item $\norm{\rho_A \rho_{B} \rho_{AB}^{-1}} \le \cC$, $\norm{\rho_{AB} \rho_A^{-1} \rho_{B}^{-1}} \le \cC$, $\norm{\rho_{ABC}\rho_B^{-1}} \le \cC$ ,
        \item $\norm{\rho_B^{-1}}\norm{\rho_B}\le \cC e^{\alpha |B|}$ ,
     \end{enumerate}
     where the constants $\cC, \alpha$ only depend on interaction strength $J$ and range $R$ of $\Phi$.
     Note that only $\rho_{ABC}$ is the Gibbs state of a local Hamiltonian, while all other states are marginals of that state.
     Furthermore, these conditions still hold if $\rho_{ABC}=\tr_{A'C'}[\rho_{A'ABCC'}]$ where $\rho_{A'ABCC'}$ is the Gibbs state on the larger system $A'ABCC'$.
\end{lemma}
\begin{proof}
Let us denote the partition function for a given interval by
\begin{equation*}
    Z_A=\Tr[e^{-H_A}].
\end{equation*}
Note that the perturbation formula \cite[Lemma 3.6]{Lenci.2005}
\begin{equation*}
    \left|\log(\Tr[e^{H+P}])-\log(\Tr[e^{H}])\right|\le\|P\|
\end{equation*}
implies that ratios of partition functions are bounded by a constant $\cC$ only depending on the interaction strength, range, and inverse temperature
\begin{equation*}
    \frac{Z_A Z_B}{Z_{AB}},\frac{Z_{AB}}{Z_AZ_B}\le \cC \, .
\end{equation*}
    For the first part, we identify $CC'$ with $B'$.
    We use an argument analogous to \cref{cor:Clean_approx_identity_prodmarginals}.
    \begin{align*}
        &\|\rho_A\rho_{B} \rho^{-1}_{AB}\|\\
         & \hspace{1cm}= \Big\|\tr_{A'BB'}[\rho^{A'}\otimes \rho^{BB'} E_{A',ABB'} E_{A,BB'}] \tr_{A'AB'}[\rho^{A'A}\otimes \rho^{B'} E_{A'A,BB'} E_{B,B'}]E_{A,B}^{-1}  \\
        &\hspace{1.8cm}  \tr_{A'B'}[\rho^{A'}\otimes \rho^{B'} E_{A',ABB'} E_{AB,B'}]^{-1}\Big\|\frac{Z_{A'A}Z_{BB'}}{Z_{A'ABB'}} \, .
    \end{align*}
    Using submultiplicativity of the norm, the first two partial traces are bounded by contractivity of the conditional expectation and \cref{prop:estimates_expansionals}.
    The same proposition also bounds $\|E_{A,B}^{-1}\|$.
    For the inverse partial trace we use \cref{Lemma:Modified_BoundingPartialTraceInverses} for the appropriately identified subsystems. Analogously,
    \begin{align*}
        &\|\rho_{AB}\rho_A^{-1}\rho_{B}^{-1}\|\\
        &\hspace{1cm}=\Big\|\tr_{A'B'}[\rho^{A'}\otimes \rho^{B'} E_{A',ABB'} E_{AB,B'}]E_{A,B}\tr_{A'BB'}[\rho^{A'}\otimes \rho^{BB'} E_{A',ABB'} E_{A,BB'}]^{-1}\\
        &\hspace{1.5cm} \tr_{A'AB'}[\rho^{A'A}\otimes \rho^{B'} E_{A'A,BB'} E_{B,B'}] ^{-1}\Big\|\frac{Z_{A'ABB'}}{Z_{A'A}Z_{BB'}} \, ,
    \end{align*}
    and, switching the notation back to a Gibbs state on the original $A'ABCC'$,
    \begin{align*}
        &\|\rho_{ABC}\rho_B^{-1}\|\\
        &\hspace{1cm}=\Big\|\tr_{A'C'}[\rho^{A'}\otimes\rho^{C'}E_{A',ABCC'}E_{ABC,C'}]E_{A,BC}E_{B,C} \frac{e^{-H_A}}{Z_A}\frac{e^{-H_C}}{Z_C}\\
        &\hspace{1.5cm}\tr_{A'ACC'}[\rho^{A'A}\otimes \rho^{CC'}E_{A'A,BCC'}E_{B,CC'}]^{-1}\Big\|\frac{Z_{A'}Z_AZ_CZ_{C'}}{Z_{AA'}Z_{CC'}}
    \end{align*}
    are bounded in the same way, apart from the last equation, which also uses $\|\rho^A\|,\|\rho^C\|\le1$.

    The second point follows from similar estimates, but no partition functions appear due to the cancellation from the two norms.
    \begin{equation*}
        \norm{\rho_B^{-1}}\norm{\rho_B}=\left\|(\rho^B)^{-1}\tr_{AC}[E_{A,BC}E_{B,C}\rho^A\otimes \rho^C]^{-1}\right\|\left\|(\rho^B)\tr_{AC}[E_{A,BC}E_{B,C}\rho^A \otimes \rho^C]\right\|
    \end{equation*}
    Again, employing \cref{Lemma:Modified_BoundingPartialTraceInverses} for the norms of inverses of partial traces, we are left with
    \begin{equation*}
        \norm{\left(\rho^B\right)^{-1}}\norm{\rho^B}=\norm{e^{H_B}}\norm{e^{-H_B}}\le e^{2\|H_B\|}
    \end{equation*}
    which yields the claim choosing $\alpha=2J$.
    
\end{proof}

We are now set to prove the main result of the section.
\begin{theorem}\label{theorem:superexponential-decay-BS-CMI}
    For $\Phi$ a finite-range, translation-invariant interaction over $\Z$, there exists a positive function $\ell \mapsto \epsilon(\ell)$ with superexponential decay in $\ell$, depending only on $J$, $R$ and $\beta$ such that for every $\Lambda \Subset \Z$ split into consecutive subintervals $\Lambda = A'ABCC'$, with $A'$ and $C'$ possibly empty, for the marginal on $ABC$ of its local Gibbs state $\tr_{A'C'}[\rho^{\Lambda}] = \rho_{ABC}$ it holds that
    \begin{equation}\label{eq:superexponential-decay-BS-CMI}
        \widehat I^{x}_{\rho} (A\mathbin{;}C|B) \le c e^{\alpha|A|} \epsilon(|B|) \quad x \in \{\operatorname{os}, \operatorname{ts}, \operatorname{rev}\} \, , 
    \end{equation}
    Here $c$ and $\alpha$ are constants only depending on inverse temperature $\beta$, strength $J$ and range $R$ of $\Phi$.
\end{theorem}
\begin{proof}
    We begin with the one-sided version and note that we can write 
    \begin{equation*}
        \widehat I^{\operatorname{os}}_\rho (A\mathbin{;}C|B) = \widehat D(\rho_{ABC} \Vert \pi_A \otimes \rho_{BC}) - \widehat D(\cE(\rho_{ABC}) \Vert \cE(\pi_A \otimes \rho_{BC}))
    \end{equation*}
    with $\cE(\cdot) := \tr_C[\cdot] \otimes \pi_C$ a conditional expectation. Using \cref{theorem:upper-bound-DPI-of-BS-entropy} we obtain
    \begin{equation*}
        \widehat I^{\operatorname{os}}_\rho (A\mathbin{;}C|B) \le \norm{\rho_{ABC}^{-1/2} \rho_{BC} \rho_{ABC}^{-1/2}} \norm{\rho_{ABC}}_1 \norm{\rho_{AB}^{1/2}} \norm{\rho_{AB}^{-1/2}}\norm{\rho_B^{-1} \rho_{AB}} \norm{\rho_{ABC}\rho_{BC}^{-1} \rho_B \rho_{AB}^{-1} - \1}
    \end{equation*}
    where we already simplified terms and cancelled constants. The fact that $\norm{AB} \le \norm{BA}$ for normal $AB$ \cite[Proposition IX.1.1]{Bhatia.1996} and that for quantum states $\norm{\rho^{-p}} =\norm{\rho^{-1}}^p$, $p \in [0, \infty)$, gives us
    \begin{align*}
        \widehat I^{\operatorname{os}}_\rho (A\mathbin{;}C|B) &\le \norm{\rho_{BC} \rho_{ABC}^{-1}} \left(\norm{\rho_{AB}^{-1}}\norm{\rho_{AB}}\right)^{1/2} \norm{\rho_B^{-1} \rho_{AB}} \norm{\rho_{ABC}\rho_{BC}^{-1} \rho_B \rho_{AB}^{-1} - \1}\\
        &\le(\norm{\rho_A^{-1}}\|\rho_A\|) \norm{\rho_A \rho_{BC} \rho_{ABC}^{-1}}\left(\norm{\rho_{AB}^{-1}}\norm{\rho_{AB}}\right)^{1/2} \\
        &\hspace{2cm}\cdot \norm{\rho_A^{-1}\rho_B^{-1} \rho_{AB}} \norm{\rho_{ABC}\rho_{BC}^{-1} \rho_B \rho_{AB}^{-1} - \1}\\
        &\le c e^{\alpha(|A| + |B|)} \varepsilon(|B|) \, . 
    \end{align*}
    In the last inequality, we used \cref{lemma:approximate-factorization-of-Gibbs-state} with \cref{Lemma:DPI_mixing_condition} and \cref{lemma:norm-estimates-of-functions-of-Gibbs-states}.
    
    For the two-sided definition, we again use \cref{theorem:upper-bound-DPI-of-BS-entropy} with the same conditional expectation and $X = \rho_{ABC}$ but with $Y = \rho_A \otimes \rho_{BC}$ instead. Applying similar simplifications as before, \cref{lemma:approximate-factorization-of-Gibbs-state}, \cref{Lemma:DPI_mixing_condition}, and \cref{lemma:norm-estimates-of-functions-of-Gibbs-states}, we obtain
    \begin{align*}
        \widehat I^{\operatorname{ts}}_\rho (A\mathbin{;}C|B) &\le \norm{\rho_A \rho_{BC} \rho_{ABC}^{-1}} \left(\norm{\rho_{AB}^{-1}}\norm{\rho_{AB}}\right)^{1/2} \norm{\rho_A^{-1} \rho_B^{-1} \rho_{AB}}\norm{\rho_{ABC}\rho_{BC}^{-1} \rho_B \rho_{AB}^{-1} - \1}\\
        &\le c e^{\alpha(|A| + |B|)} \varepsilon(|B|) \, . 
    \end{align*}
    
    For the reversed version, we set $X = \pi_A \otimes \rho_{BC}$, $Y = \rho_{ABC}$ with the conditional expectation fixed as before. Employing \cref{theorem:upper-bound-DPI-of-BS-entropy} in the form of the second inequality we obtain after initial simplifications
    \begin{align*}
        \widehat I^{\operatorname{rev}}_\rho(A\mathbin{;}C|B) &\le \norm{\rho_{BC}^{-1/2} \rho_{ABC} \rho_{BC}^{-1/2}} \left(\norm{\rho_{B}^{-1}}\norm{\rho_B}\right)^{1/2} \norm{\rho_{ABC}^{-1} \rho_{BC}} \norm{\rho_{ABC}\rho_{BC}^{-1}\rho_B \rho_{AB}^{-1} - \1} \\
        &\le \norm{\rho_{ABC} \rho_A^{-1} \rho_{BC}^{-1}} \left(\norm{\rho_{B}^{-1}}\norm{\rho_B}\right)^{1/2} \norm{\rho_A^{-1}}\norm{\rho_A} \\
        &\hspace{2cm} \cdot \norm{\rho_A \rho_{BC} \rho_{ABC}^{-1}} \norm{\rho_{ABC} \rho_{BC}^{-1} \rho_B \rho_{AB}^{-1} - \1}\\
        &\le c e^{\alpha(|A| + |B|)} \varepsilon(|B|) \, . 
    \end{align*}
    The last estimation again follows from \cref{lemma:approximate-factorization-of-Gibbs-state}, \cref{Lemma:DPI_mixing_condition} and \cref{lemma:norm-estimates-of-functions-of-Gibbs-states}. 

    Ultimately, for all bounds, we absorbed the exponential growth in the $B$ system into the superexponential decay, yielding the claimed result.
\end{proof}

Note that all the BS-CMI bounds depend on the size of one of the side systems exponentially. At least in the case of the two-sided version, one can relatively easily show a \emph{constant} upper bound on the DPI independent of the dimension, which might hint towards possible improvements. Lastly, we would also like to restate \cref{rem:complementing-perspective-BS-DPI} and translate it to the context of one-sided BS-CMI to make the inequalities more accessible. 

\begin{remark}
    In the case where $\rho_{ABC}$ is a state on a tripartite system (not necessarily with any relation to a Gibbs state), the following chain of inequalities holds:
    \begin{align*}
        &\left(\frac{\pi}{8}\right)^4 \norm{\rho_{ABC}^{-1/2}\rho_{BC} \rho_{ABC}^{-1/2}}^{-4} \norm{\rho_{AB}^{-1}}^{-2}\norm{\rho_B \rho_{AB}^{-1}\rho_{ABC} - \rho_{BC}}\\
        &\hspace{4.2cm}\le \widehat I^{\operatorname{os}}_\rho (A\mathbin{;}C|B) \le \\
        &\hspace{4.8cm}\norm{\rho_{ABC}^{-1/2}\rho_{BC} \rho_{ABC}^{-1/2}} (\norm{\rho_{AB}^{-1}}\norm{\rho_{AB}})^{1/2} \norm{\rho_{B}^{-1} \rho_{AB}} \norm{\rho_{ABC}\rho_{BC}^{-1}\rho_{AB}\rho_B^{-1}  - \1} \, . 
    \end{align*}
\end{remark}
\begin{remark}
    An analogue of \cref{theorem:superexponential-decay-BS-CMI} for exponentially-decaying interactions is provided in \cref{theorem:superexponential-decay-BS-CMI_shortrange}. The main difference is the exponential instead of the superexponential decay with $|B|$, which only holds for $\beta$ small enough.
\end{remark}

\subsection{Approximate factorisation of the purity}
\label{sec:purity}
 
We now consider another possible measure of conditional independence given by the approximate factorisation of the purity. This notion is inspired by \cite{Vermersch.2023}, where it was shown that for a $\rho_\Lambda$ prepared by a finite depth circuit with $\Lambda \Subset \mathbb{Z}$ split as $\Lambda = ABC$ it holds that
\begin{equation} \label{eq:exactpurity}
    \frac{\Tr_{AB}[\rho_{AB}^2]\Tr_{BC}[\rho_{BC}^2]}{\Tr_\Lambda[\rho_\Lambda^2] \Tr_B[\rho_B^2] } = 1
\end{equation}
whenever $|B|\geq 2 \ell -1$ with $\ell$ being the depth of the circuit. This may suggest a definition of conditional independence based on a notion of CMI defined in terms of R\'enyi-$2$ entropies, as, in analogy with Eq. \eqref{eq:CMI1},
\begin{equation}\label{eq:notCMI2}
    I_2(A:C \vert B) :=   \log  \frac{\Tr_\Lambda[\rho_\Lambda^2] \Tr_B[\rho_B^2] }{\Tr_{AB}[\rho_{AB}^2]\Tr_{BC}[\rho_{BC}^2]} = S_2(\rho_{AB}) + S_2(\rho_{BC}) - S_2(\rho_B) - S_2(\rho_\Lambda) \, ,
\end{equation}
with $S_2(\rho)=-\log \Tr[\rho^2]$. It should not be confused with the $2$-CMI that could be defined from Petz Rényi or sandwiched Rényi divergences.
However, notice that this quantity, contrary to $I(A:C \vert B)$, $\IosBS{\rho}{A}{C}{B}$, $\ItsBS{\rho}{A}{C}{B}$ and $\IrevBS{\rho}{A}{C}{B}$, is not necessarily positive, and likely lacks most other relevant information-theoretic properties.

Nevertheless, motivated by the problem of efficiently learning R\'enyi entropies the authors of \cite{Vermersch.2023} introduce an approximate factorisation condition for this measure, which holds when the previous Eq. \eqref{eq:exactpurity} fails up to, at most, an exponentially small error in $|B|$. They proved that translation-invariant Matrix Product Density Operators satisfy this property, conjectured it for a larger class of states, and numerically verified it for some relevant models. This approximate factorisation is equivalent to the exponential decay in $\vert B \vert$ of Eq. \eqref{eq:notCMI2}. In this section we show this property for any translation-invariant, finite-range Hamiltonian in 1D at any inverse temperature $\beta >0$. 

\begin{proposition}\label{th:purity}
    Let $\Phi$ be a finite-range translation-invariant interaction over $\mathbb{Z}$. Then, there exist positive constants $c_p, \alpha_p$ depending only on the strength $J$, range of the interaction $R$ and inverse temperature $\beta>0$ with the following property: For every $\Lambda \Subset \mathbb{Z}$ split as $\Lambda = ABC$, where $B$ shields $A$ from $C$  (see Figure \ref{fig:1}), and for $\rho_\Lambda := \rho^\Lambda$ the Gibbs state on $\Lambda$,
    \begin{equation}\label{eq:approximate_factorization}
         \left| \frac{\Tr_{AB}[\rho_{AB}^2]\Tr_{BC}[\rho_{BC}^2]}{\Tr_\Lambda[\rho_\Lambda^2] \Tr_B[\rho_B^2] } -1\right| \leq c_p e^{- \alpha_p |B|} \, .
    \end{equation}
\end{proposition}
\begin{proof}
    Consider $\ell \in \mathbb{N}$ such that $|B|\geq 3 \ell$ and split $B$ into $B_1$, $B_2$ and $B_3$ such that $|B_1|, |B_2|, |B_3| \geq \ell $ as in Figure \ref{fig:1}.
    \begin{figure}[h]
        \begin{center}
            \begin{tikzpicture}[scale=0.68]
                \definecolor{frenchblue}{rgb}{0.0, 0.45, 0.73}
                \Block[5,darkblue!60!white,A,1];
                \draw [thick,darkblue,decorate,decoration={brace,amplitude=10pt,mirror},xshift=-0.5pt,yshift=-0.6pt](-0.5,-0.6) -- (4.5,-0.6) node[black,midway,yshift=-0.7cm] { \textcolor{darkblue}{$A$}};
                \begin{scope}[xshift=5cm]
                    \Block[4,midblue!50!cyan!80!white,B,1];
                    \draw [thick,midblue!50!cyan!80!white,decorate,decoration={brace,amplitude=10pt,mirror},xshift=-0.5pt,yshift=-0.6pt](-0.5,-0.6) -- (3.5,-0.6) node[black,midway,yshift=-0.7cm] { \textcolor{midblue!50!cyan!80!white}{$B_1$}};
                    \draw [thick,cyan,decorate,decoration={brace,amplitude=10pt,mirror},xshift=-0.5pt,yshift=-0.6pt](-0.5,-2) -- (11.5,-2) node[black,midway,yshift=-0.7cm] { \textcolor{cyan}{$B$}};
                \end{scope}
                \begin{scope}[xshift=9cm]
                    \Block[4,cyan!80!white,C,1];
                    \draw [thick,cyan!60!white,decorate,decoration={brace,amplitude=10pt,mirror},xshift=-0.5pt,yshift=-0.6pt](-0.5,-0.6) -- (3.5,-0.6) node[black,midway,yshift=-0.7cm] { \textcolor{cyan!60!white}{$B_2$}};
                \end{scope}
                \begin{scope}[xshift=13cm]
                    \Block[4,lightgreen!30!cyan!60!white,C,1];
                    \draw [thick,lightgreen!30!cyan!60!white,decorate,decoration={brace,amplitude=10pt,mirror},xshift=-0.5pt,yshift=-0.6pt](-0.5,-0.6) -- (3.5,-0.6) node[black,midway,yshift=-0.7cm] { \textcolor{lightgreen!30!cyan!60}{$B_3$}};
                \end{scope}
                \begin{scope}[xshift=17cm]
                    \Block[5,midgreen!60!white,D,1];
                    \draw [thick,midgreen,decorate,decoration={brace,amplitude=10pt,mirror},xshift=-0.5pt,yshift=-0.6pt](-0.5,-0.6) -- (4.5,-0.6) node[black,midway,yshift=-0.7cm] { \textcolor{midgreen}{$C$}};
                \end{scope}
                \node at (23,0.8) {\huge $\Lambda$};
            \end{tikzpicture} 
        \end{center}
        \caption{Representation of an interval $\Lambda$ split into three subintervals $\Lambda=ABC$, with $B$ further split into $B_1$, $B_2$ and $B_3$ such that $|B_1|, |B_2|, |B_3| \geq \ell $.}
        \label{fig:1}
    \end{figure}
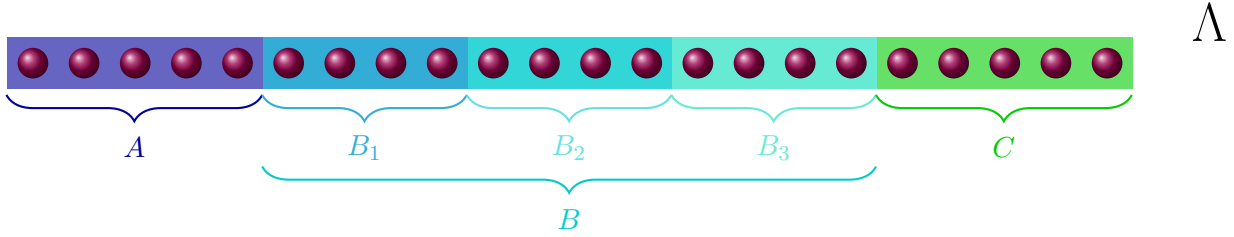
    For any $X \Subset \Lambda$, let us denote $Z_X:= \Tr_X[e^{-H_X}]$ and by $\rho^X$ the Gibbs state on $X$, i.e. $\rho^X:= e^{-H_X}/Z_X$ (notice the difference with the marginal of the global state $\rho_X$). Let us further define 
    \begin{equation}\label{eq:lambda}
        \lambda_{ABC} := \frac{Z_{ABC}Z_B}{Z_{AB}Z_{BC}} \, .
    \end{equation}
    Then, we can rewrite $\lambda_{ABC}$ in terms of expansionals as
    \begin{align*}
        \lambda_{ABC} &= \frac{\Tr_{ABC}[e^{-H_{ABC}}] \Tr_{B}[e^{-H_{B}}]}{\Tr_{AB}[e^{-H_{AB}}] \Tr_{BC}[e^{-H_{BC}}]} \\[1mm]
        & = \big(\Tr_{ABC}[\rho^{ABC} \underbrace{e^{H_{ABC}} e^{-H_A-H_{BC}}}_{E_{A,BC}^{-1 \, *}}]\big)^{-1}  \Tr_{AB}[\rho^{AB} \underbrace{e^{H_{AB}} e^{-H_A-H_{B}}}_{E_{A,B}^{-1 \, *}}] \, .
    \end{align*}
    By \Cref{prop:BoundingPartialTraceInverses}, we have 
    \begin{equation}\label{eq:uni_bound_lambda}
        |\lambda_{ABC}|, |\lambda_{ABC}^{-1}| \leq \mathcal{C}^2 \, ,
    \end{equation}
    and by Step 2 in the proof of Proposition 8.1 of \cite{Bluhm.2022}, there exist $\tilde c_1, \alpha_1 > 0$ such that 
    \begin{equation}\label{eq:uni_bound_lambda-1}
        |\lambda_{ABC} -1 | \leq \tilde c_1 e^{-\alpha_1 \ell} \, .
    \end{equation}
    Note that the above two estimates can be used to obtain
    \begin{equation}
        |\lambda_{ABC}^{-1} - 1| = |\lambda_{ABC}^{-1}||1 - \lambda_{ABC}| \le \cC^2 \tilde c_1 e^{-\alpha_1 \ell} =: c_1 e^{-\alpha_1 \ell}
    \end{equation}
    We will follow similar but more technical steps for  
    \begin{equation}\label{eq:lambda_hat}
        \widehat{\lambda}_{ABC}^{-1} := \frac{\Tr_{AB}[\rho_{AB}^2]\Tr_{BC}[\rho_{BC}^2]}{\Tr_\Lambda[\rho_\Lambda^2] \Tr_B[\rho_B^2] } \, .
    \end{equation}
    Let us analyse each of the terms involved independently. For that, we will consider and compare the partition functions and Gibbs states associated to $H_\Lambda$ and to $2H_\Lambda$, denoting $\widetilde{Z}_X := \Tr_X[e^{-2H_X}]$ and $\widetilde{\rho}_X:= e^{-2H_X}/\widetilde{Z}_X$ for $X \Subset \Lambda$.
    \begin{itemize}
        \item For $\Tr_{\Lambda}[\rho_{\Lambda}^2]$, we have:
        \begin{align*}
             \Tr_{\Lambda}[\rho_{\Lambda}^2] & = \frac{\Tr_{ABC}\big[e^{-2H_{ABC}}\big]}{Z_{ABC}^2}  = \frac{\widetilde{Z}_{ABC}}{Z_{ABC}^2} \, .
        \end{align*}
        \item For $\Tr_{AB}[\rho_{AB}^2]$, we obtain:
        \begin{align*}
            \Tr_{AB}[\rho_{AB}^2] & = \frac{1}{Z_{ABC}^2}\Tr_{AB}\big[\tr_C[e^{-H_{ABC}}]^2\big] \\
            & = \frac{Z_C^2}{Z_{ABC}^2} \Tr_{AB}\Big[ e^{-H_{AB}} \tr_C\big[ \rho^C \underbrace{e^{H_{AB}} e^{H_{C}} e^{-H_{ABC}}}_{E_{AB,C}^*}\big] \tr_C \big[ \underbrace{e^{-H_{ABC}}  e^{H_{C}} e^{H_{AB}}}_{E_{AB,C}}  \rho^C\big] e^{-H_{AB}}\Big] \\
            & = \frac{\widetilde{Z}_{AB} Z_C^2}{Z_{ABC}^2} \Tr_{AB} \Big[ \widetilde{\rho}^{AB} \tr_C \big[ \rho^C {E_{AB,C}^*}\big] \tr_C \big[ E_{AB,C} \rho^C \big] \Big] \, .
        \end{align*}
        \item Similarly for $\Tr_{BC}[\rho_{BC}^2]$, we can write:
        \begin{align*}
            \Tr_{BC}[\rho_{BC}^2] & = \frac{Z_A^2 \widetilde{Z}_{BC} }{Z_{ABC}^2} \Tr_{BC}\Big[ \widetilde{\rho}^{BC} \tr_A \big[ \rho^A {E_{A,BC}^*} \big] \tr_A \big[ E_{A,BC} \rho^A \big] \Big] \, .
        \end{align*}
        \item Finally, for $\Tr_{B}[\rho_{B}^2]$, we have:
        \begin{align*}
            \Tr_{B}[\rho_{B}^2] & = \frac{1}{Z_{ABC}^2}\Tr_{B}[\tr_{AC}[e^{-H_{ABC}}]^2] \\
            & = \frac{Z_A^2 Z_C^2}{Z_{ABC}^2} \Tr_{B} \Big[ e^{-H_{B}} \tr_{AC} \big[ \rho^A \otimes \rho^C \underbrace{e^{H_{A}} e^{H_{B}} e^{-H_{AB}}}_{E_{A,B}^*} \underbrace{e^{H_{AB}} e^{H_{C}} e^{-H_{ABC}}}_{E_{AB,C}^*} \big]  \\
            & \phantom{asasassdAasd}  \tr_{AC}\big[ \underbrace{e^{-H_{ABC}}  e^{H_{AB}} e^{H_{C}} }_{E_{AB,C}} \underbrace{ e^{-H_{AB}} e^{H_{A}} e^{H_{B}} }_{E_{A,B}} \rho^A \otimes \rho^C \big] e^{-H_{B}} \Big] \\
            & = \frac{Z_A^2 Z_C^2 \widetilde{Z}_B }{Z_{ABC}^2} \Tr_{B} \Big[ \widetilde{\rho}^B \tr_{AC} \big[ \rho^A \otimes \rho^C {E_{A,B}^*} {E_{AB,C}^*}\big]  \tr_{AC} \big[ {E_{AB,C}} {E_{A,B}} \rho^A \otimes \rho^C \big]  \Big]\, .
        \end{align*}
    \end{itemize}
    Denoting 
    \begin{equation}
        \widetilde{\lambda}_{ABC} := \frac{\widetilde{Z}_{ABC}\widetilde{Z}_{B} }{\widetilde{Z}_{AB}\widetilde{Z}_{BC} } 
    \end{equation}
    and replacing all these values in Eq. \eqref{eq:lambda_hat}, after noticing that all partition functions $Z_X$ of the original Hamiltonian cancel, we obtain
    \begin{equation*}
        \widehat{\lambda}_{ABC}^{-1} =  \widetilde{\lambda}_{ABC}^{-1} \frac{ \Tr_{AB} \Big[ \widetilde{\rho}^{AB} \tr_C \big[ \rho^C {E_{AB,C}^*}\big] \tr_C \big[ E_{AB,C} \rho^C \big] \Big] \Tr_{BC}\Big[ \widetilde{\rho}^{BC} \tr_A \big[ \rho^A {E_{A,BC}^*} \big] \tr_A \big[ E_{A,BC} \rho^A \big] \Big]}{\Tr_{B} \Big[ \widetilde{\rho}^B \tr_{AC} \big[ \rho^A \otimes \rho^C {E_{A,B}^*} {E_{AB,C}^*}\big]  \tr_{AC} \big[ {E_{AB,C}} {E_{A,B}} \rho^A \otimes \rho^C \big]  \Big]} \, .
    \end{equation*}
    Let us denote the last fraction above by $\chi_{ABC}$. Then, it is clear that 
    \begin{align*}
        \left| \frac{\Tr_{AB}[\rho_{AB}^2]\Tr_{BC}[\rho_{BC}^2]}{\Tr_\Lambda[\rho_\Lambda^2] \Tr_B[\rho_B^2] } -1\right|  & = \left| \widetilde{\lambda}_{ABC}^{-1} \chi_{ABC} - 1 \right| \\
        & \leq \left|\widetilde{\lambda}_{ABC}^{-1} - 1 \right| + \left| \widetilde{\lambda}_{ABC}^{-1} \right| \left|  \chi_{ABC} - 1 \right| \\
        & \leq c_1 e^{-\alpha_1 \ell} + \mathcal{C}^2 \left|  \chi_{ABC} - 1 \right| \, ,
    \end{align*}
    where we are using Eq. \eqref{eq:uni_bound_lambda} and Eq. \eqref{eq:uni_bound_lambda-1} for the partition functions associated to the $2H_\Lambda$. Note that the constants would change slightly, just modifying their dependence from $\beta$ to $\beta/2$, but we keep the same notation for simplicity. Finally, we bound the last term in the expression above. We will repeat the combination of \Cref{prop:estimates_expansionals} and \Cref{theorem:local_indistinguishability}. 

    Let us consider a splitting of $B$ into $B_1 $, $B_2$ and $B_3$ such that $|B_1|, |B_2|, |B_3| \geq \ell$ (see Figure \ref{fig:1}). First, by triangle inequality, Hölder's inequality, point (ii) of \Cref{prop:estimates_expansionals} and \Cref{prop:BoundingPartialTraceInverses}, we get
    \begin{equation}\label{eq:bound_tr_AB}
        \begin{aligned}
            &\left|\Tr_{AB} \Big[ \widetilde{\rho}^{AB} \tr_C \big[ \rho^C {E_{AB,C}^*}\big] \tr_C \big[ E_{AB,C} \rho^C \big] \Big] -  \Tr_{AB} \Big[ \widetilde{\rho}^{AB} \tr_C \big[ \rho^C {E_{B_3,C}^*}\big] \tr_C \big[ E_{B_3,C} \rho^C \big] \Big] \right| \\
            & \hspace{1cm} = \left|\Tr_{AB} \Big[ \widetilde{\rho}^{AB}  \tr_C \big[ \rho^C\big(  {E_{AB,C}^*} -  {E_{B_3,C}^*}  \big)\big] \tr_C \big[ E_{B_3,C} \rho^C \big] \Big] \right.\\
            &\hspace{1.5cm}  \left. + \Tr_{AB} \Big[ \widetilde{\rho}^{AB}  \tr_C \big[ \rho^C  {E_{AB,C}^*} \big] \tr_C \big[\left( E_{AB,C} -  {E_{B_3,C}}  \right) \rho^C \big] \Big] \right| \\
            &\hspace{1cm} \leq \mathcal{C} \norm{{E_{AB,C}^*} -  {E_{B_3,C}^*}} + \mathcal{C} \norm{{E_{AB,C}} -  {E_{B_3,C}} }  \\
            &\hspace{1cm} \leq 2 \mathcal{C} \delta( \ell) \, . 
        \end{aligned}
    \end{equation}
    Similarly for the analogous term tracing out $BC$, 
    \begin{equation}\label{eq:bound_tr_BC}
        \left|\Tr_{BC} \Big[ \widetilde{\rho}^{BC} \tr_A \big[ \rho^A {E_{A,BC}^*}\big] \tr_A \big[ E_{A,BC} \rho^A \big] \Big] - \Tr_{BC} \Big[ \widetilde{\rho}^{BC} \tr_A \big[ \rho^A {E_{A,B_1}^*}\big] \tr_A \big[ E_{A,B_1} \rho^A \big] \Big] \right| \leq  2 \mathcal{C} \delta( \ell) \, . 
    \end{equation}
    The term tracing out $B$ is slightly more involved but follows the same lines:
    \begin{equation}\label{eq:bound_tr_B}
        \begin{aligned}
            &\left|\Tr_{B} \Big[ \widetilde{\rho}^B \tr_{AC} \big[ \rho^A \otimes \rho^C {E_{A,B}^*} {E_{AB,C}^*}\big]  \tr_{AC} \big[ {E_{AB,C}} {E_{A,B}} \rho^A \otimes \rho^C \big]  \Big] \right. \\
            & \hspace{1cm} \left. - \Tr_{B} \Big[ \widetilde{\rho}^B \tr_{AC} \big[ \rho^A \otimes \rho^C {E_{A,B_1}^*} \otimes  {E_{B_3,C}^*}\big]  \tr_{AC} \big[ {E_{A,B_1}} \otimes {E_{B_3,C}} \rho^A \otimes \rho^C \big]  \Big]\right|\\
            & = \left|\Tr_{B} \Big[ \widetilde{\rho}^B \tr_{AC} \big[ \rho^A \otimes \rho^C {E_{A,B}^*} {E_{AB,C}^*}\big]  \tr_{AC} \big[ {E_{AB,C}} {E_{A,B}} \rho^A \otimes \rho^C \big]  \Big] \right. \\
            & \hspace{1cm} \left. - \Tr_{B} \Big[ \widetilde{\rho}^B \tr_{A} \big[ \rho^A  {E_{A,B_1}^*} \big] \tr_{A} \big[{E_{A,B_1}} \rho^A \big]  \otimes   \tr_{C} \big[ \rho^C  {E_{B_3,C}^*}\big] \tr_{C} \big[ {E_{B_3,C}} \rho^C \big]  \Big] \right| \\
            & \hspace{1cm} \le  4 \mathcal{C}^3 \delta(\ell)  \, . 
        \end{aligned}
    \end{equation}
    Additionally, note that by \cref{Lemma:Modified_BoundingPartialTraceInverses}
    \begin{equation}\label{eq:constant_bound_tr_B}
        \begin{aligned}
            \Tr_{B} \Big[ \widetilde{\rho}^B \tr_{AC} &\big[ \rho^A \otimes \rho^C {E_{A,B}^*} {E_{AB,C}^*}\big]  \tr_{AC} \big[ {E_{AB,C}} {E_{A,B}} \rho^A \otimes \rho^C \big]  \Big] \\
            &\ge \norm{(\tr_{AC} \big[ \rho^A \otimes \rho^C {E_{A,B}^*} {E_{AB,C}^*}\big]  \tr_{AC} \big[ {E_{AB,C}} {E_{A,B}} \rho^A \otimes \rho^C \big])^{-1}}^{-1} \ge \cC^{-2}
        \end{aligned}
    \end{equation}  
    Now, combining Eq. \eqref{eq:bound_tr_AB}, Eq. \eqref{eq:bound_tr_BC}, Eq. \eqref{eq:bound_tr_B} and Eq. \eqref{eq:constant_bound_tr_B}, we have 
    \begin{align*}
        &\left|  \chi_{ABC} - 1 \right| \\
        & \hspace{0.5cm} \leq \mathcal{C}^2  \left| \Tr_{B} \Big[ \widetilde{\rho}^B \tr_{AC} \big[ \rho^A \otimes \rho^C {E_{A,B}^*} {E_{AB,C}^*}\big]  \tr_{AC} \big[ {E_{AB,C}} {E_{A,B}} \rho^A \otimes \rho^C \big]  \Big] \right. \\
        & \hspace{1.5cm} - \left. \Tr_{AB} \Big[ \widetilde{\rho}^{AB} \tr_C \big[ \rho^C {E_{AB,C}^*}\big] \tr_C \big[ E_{AB,C} \rho^C \big] \Big] \Tr_{BC}\Big[ \widetilde{\rho}^{BC} \tr_A \big[ \rho^A {E_{A,BC}^*} \big] \tr_A \big[ E_{A,BC} \rho^A \big] \Big] \right| \\
        &\hspace{0.5cm} \leq \mathcal{C}^2  \left| \Tr_{B} \Big[ \widetilde{\rho}^B \tr_{A} \big[ \rho^A  {E_{A,B_1}^*} \big] \tr_{A} \big[  {E_{A,B_1}} \rho^A \big]  \otimes   \tr_{C} \big[ \rho^C  {E_{B_3,C}^*}\big]\tr_{C} \big[ {E_{B_3,C}} \rho^C \big]  \Big]  \right. \\
        &\hspace{1.5cm} - \left. \Tr_{AB} \Big[ \widetilde{\rho}^{AB} \tr_C \big[ \rho^C {E_{B_3,C}^*}\big] \tr_C \big[ E_{B_3,C} \rho^C \big] \Big] \Tr_{BC} \Big[ \widetilde{\rho}^{BC} \tr_A \big[ \rho^A {E_{A,B_1}^*}\big] \tr_A \big[ E_{A,B_1} \rho^A \big] \Big] \right| \\
        &\hspace{1.5cm}+ 4 \mathcal{C}^5 \delta(\ell) +  2 \mathcal{C}^5 \delta( \ell)+  2 \mathcal{C}^5 \delta( \ell)    \, .
    \end{align*}
    To conclude, we will estimate the difference in the previous expression using the local indistinguishability of Gibbs states as in \Cref{theorem:local_indistinguishability}. Indeed, note that
    \begin{align*}
        & \left|\Tr_{BC} \Big[ \widetilde{\rho}^{BC} \tr_A \big[ \rho^A {E_{A,B_1}^*}\big] \tr_A \big[ E_{A,B_1} \rho^A \big] \Big] - \Tr_{B} \Big[ \widetilde{\rho}^{B} \tr_A \big[ \rho^A {E_{A,B_1}^*}\big] \tr_A \big[ E_{A,B_1} \rho^A \big] \Big]\right|\le \mathcal{C}^2 \,  c_2  \, e^{-\alpha_2 \ell} \, ,
    \end{align*}
    as well as 
    \begin{align*}
       &\left|\Tr_{AB} \Big[ \widetilde{\rho}^{AB} \tr_C \big[ \rho^C {E_{B_3,C}^*}\big] \tr_C \big[ E_{B_3,C} \rho^C \big] \Big] - \Tr_{B} \Big[ \widetilde{\rho}^{B} \tr_C \big[\rho^C {E_{B_3,C}^*}\big] \tr_C \big[ E_{B_3,C} \rho^C \big] \Big]\right| \leq \mathcal{C}^2 \,  c_2  \, e^{-\alpha_2 \ell} \, ,
    \end{align*}
    for certain constants $c_2, \alpha_2 >0$. Therefore, denoting 
    \begin{equation*}
        \xi_{B_1} := \tr_A \big[ \rho^A {E_{A,B_1}^*}\big] \tr_A \big[ E_{A,B_1} \rho^A \big] \; , \quad  \xi_{B_3} := \tr_C \big[ \rho^C {E_{B_3,C}^*}\big] \tr_C \big[ E_{B_3,C} \rho^C \big] \, ,
    \end{equation*}
    we have
    \begin{align*}
        & \left|  \Tr_{B} \big[ \widetilde{\rho}^B \xi_{B_1}  \otimes   \xi_{B_3} \big] - \Tr_{AB} \big[ \widetilde{\rho}^{AB} \xi_{B_3}  \big] \Tr_{BC} \big[ \widetilde{\rho}^{BC}  \xi_{B_1}  \big] \right|  \\
        & \hspace{5cm} \leq  \left|  \Tr_{B} \big[ \widetilde{\rho}^B \xi_{B_1}  \otimes   \xi_{B_3} \big]  - \Tr_{B} \big[ \widetilde{\rho}^{B} \xi_{B_3}  \big] \Tr_{B} \big[ \widetilde{\rho}^{B}  \xi_{B_1}  \big] \right| +  2 \mathcal{C}^4 \, c_2  \, e^{-\alpha_2 \ell}  \, , 
    \end{align*}
    and we conclude the proof by using exponential decay of correlations as in \Cref{prop:decay_correlations} giving us an upper estimate $ \cC^{4} c_3 e^{-\alpha_3 \ell}$. Combining all of the above results finally gives the claim:
    \begin{equation*}
        |\widehat{\lambda}_{ABC}^{-1} - 1| \le c_1 e^{-\alpha_1 \ell} + 8 \cC^{7}\delta(\ell) + \cC^{8}(2 c_2 e^{-\alpha_2 \ell} + c_3 e^{-\alpha_3 \ell}) \, .
    \end{equation*}
\end{proof}

\begin{remark}
    Note that \cref{th:purity} also holds for exponentially-decaying interactions, with the same proof, just by adapting to that case the technical tools employed. This is the content of \cref{th:purity_shortrange}.  
\end{remark}

\section{Applications}
This section is devoted to applications of the main results from the previous section in the context of MPO approximations and learning of Gibbs states. In \cref{sec:MPOappox_recovchannels}, we show that a set of information-theoretic criteria, most importantly the decay of the BS-CMI, imply an efficient MPO representation of a state, which by the previous section exists for one-dimensional Gibbs states.  Using the explicit form of this reconstruction in \Cref{sec:tomography}, we show that this representation can be learned by tomography of small marginals. In \Cref{sec:MPOpurity}, we outline the scheme for estimating the global purity given Proposition \ref{th:purity}.

\subsection{Positive MPO approximations from recovery maps}\label{sec:MPOappox_recovchannels}
In this section, we provide a sequential reconstruction using a symmetric recovery map for the BS-CMI, newly introduced in this manuscript. The motivation for this map arises from Eq. \eqref{eq:asymmetric_BS-recovery}, where we recalled that the DPI for the BS-entropy saturates if, and only if, each state can be recovered from the other by the so-called (asymmetric) BS-recovery condition. This condition, albeit appealing and useful for applications in the context of Gibbs states (cf. \Cref{theorem:superexponential-decay-BS-CMI}), is operationally flawed because of its lack of positivity, let alone Hermiticity. The desire to have a related positive map encourages us to define the following \textit{symmetric recovery map} for a particular case in which both states are defined in a tripartite space and the conditional expectation considered is a partial trace:
\begin{align}\label{eq:symmentric_recovery_map}
    \cR(X) &= \rho_B^{1/2}(\rho_B^{-1/2}\rho_{AB}\rho_B^{-1/2})^{1/2}\rho_B^{-1/2}X\rho_B^{-1/2}(\rho_B^{-1/2}\rho_{AB}\rho_B^{-1/2})^{1/2}\rho_B^{1/2} \, . 
\end{align}
Surprisingly, this is a completely positive linear map with the same fixed points as the BS-recovery condition in a tripartite space, which is a trace-preserving linear map, but not even positive \cite{Bluhm.2.2004}. The definition of this map is not arbitrary, since it follows from the combination of some bounds obtained in \cite{Bluhm.2020} and \cite{Carlen.2020}, as we will see in the Lemma below. Interestingly, it yields a single-shot recovery error bound that involves the inverse BS-CMI, the lowest eigenvalue of marginals and the maximal mutual information.


\begin{lemma}\label{lem:singleRecoveryError}
  Given a tripartite state $\rho_{ABC}$, the following bound on the distance between the state and the recovery map from Eq. \eqref{eq:symmentric_recovery_map} holds
    \begin{align}\label{eq:recoveryIndividual}
        \widehat I_\rho^{\operatorname{rev}}(A\mathbin{;}C|B) &= \widehat D(\pi_A\otimes\rho_{BC}\|\rho_{ABC})-\widehat D(\pi_A\otimes\rho_{B}\|\rho_{AB})\\
        &\ge\left(\frac\pi8\right)^4\|\Gamma\|^{-2}\|\cR(\rho_{BC})-\rho_{ABC}\|_1^4 \, , 
    \end{align}
    where $\Gamma=\rho_{BC}^{-1/2}\rho_{ABC}\rho_{BC}^{-1/2}$.
\end{lemma}
\begin{remark}
    Note that $\cR(X)$ is completely positive but not trace-preserving, so it is not a quantum channel.
    Several choices of recovery map for the BS-entropy are possible, see \cite{Bluhm.2020} for an alternative definition.
    The one in the above Lemma, despite its more complicated form, is necessary for proving the bound on the Lipschitz constants below and has not been considered before to the best of our knowledge. 
\end{remark}
\begin{proof}[Proof of Lemma \ref{lem:singleRecoveryError}]
    We use a particular case of the strengthened data-processing inequality for the BS-entropy \cite[Theorem 5.3]{Bluhm.2020}, where we choose the channel to be the conditional expectation $\cE(\cdot) = \tr_C[\cdot] \otimes \pi_C$, $\sigma=\pi_A\otimes\rho_{BC}$ for the first argument, and $\rho=\rho_{ABC}$ for the second (note the different convention for the naming of arguments). Employing now \cite[Theorem 5.3]{Bluhm.2020} gives the bound
    \begin{equation}\label{eq:intermediate-strengthened-BS-DPI}
        \begin{aligned}
            \widehat{I}^{\operatorname{rev}}(A\mathbin{;}C|B) &= \widehat D(\pi_A\otimes\rho_{BC}\|\rho_{ABC})-\widehat D(\pi_A\otimes\rho_{B}\|\rho_{AB}) \\
            &\ge \left(\frac{\pi}{4}\right)^4\|\Gamma\|^{-2}\|\rho_{BC}^{1/2}\rho_{B}^{-1/2}\Gamma_{\mathcal{E}}^{1/2}\rho_B^{1/2}-\Gamma^{1/2}\rho_{BC}^{1/2}\|^4_2 \, ,
        \end{aligned}
    \end{equation}
    where $\Gamma=\rho_{BC}^{-1/2}\rho_{ABC}\rho_{BC}^{-1/2}$ and $\Gamma_{\mathcal{E}}=\rho_{B}^{-1/2}\rho_{AB}\rho_{B}^{-1/2}$. Noticing that 
    \begin{align*}
        \Tr[\cR(\rho_{BC})] & = \Tr_{ABC}[\rho_B^{1/2}(\rho_B^{-1/2}\rho_{AB}\rho_B^{-1/2})^{1/2}\rho_B^{-1/2}\rho_{BC}\rho_B^{-1/2}(\rho_B^{-1/2}\rho_{AB}\rho_B^{-1/2})^{1/2}\rho_B^{1/2}]\\
        &= \Tr_{AB}[\rho_B^{1/2}(\rho_B^{-1/2}\rho_{AB}\rho_B^{-1/2})\rho_B^{1/2}] = 1
    \end{align*}
    in addition to $\Tr_{ABC}[\rho_{ABC}] = 1$ allows us to employ \cite[Lemma 2.2]{Carlen.2020}, which states
    \begin{equation*}
        \|X^*X-Y^*Y\|_1\le2\|X-Y\|_2
    \end{equation*}
    for operators $X,Y$ such that $\Tr[X^*X]=\Tr[Y^*Y]=1$. Considering here
    \begin{equation*}
        X= \rho_{BC}^{1/2}\rho_{B}^{-1/2}\Gamma_{\mathcal{E}}^{1/2}\rho_B^{1/2} \; , \quad Y = \Gamma^{1/2}\rho_{BC}^{1/2} \, ,
    \end{equation*}
    this, in turn, gives a lower bound to Eq. \eqref{eq:intermediate-strengthened-BS-DPI} and a strengthened DPI for the inverse BS-CMI involving the distance of $\rho_{ABC}$ to $\cR(\rho_{BC})$ in $1$ norm, i.e. the claim of \cref{lem:singleRecoveryError}.
\end{proof}


As we will see in the next pages, the recovery map introduced in Eq. \eqref{eq:symmentric_recovery_map}, when applied iteratively, provides a set of information-theoretic criteria for the existence of MPO descriptions of an arbitrary state. For that, let us consider a sequence of recovery maps that can reconstruct a quantum state on a spin chain.
We assume an appropriately subdivided chain on subintervals $A_1,\ldots,A_N$ as in Figure \ref{fig:3}.


\begin{figure}[h]
    \begin{center}
        \begin{tikzpicture}[scale=0.7]
            \Block[3,darkblue!60!white,A,1];
            \draw [thick,darkblue,decorate,decoration={brace,amplitude=10pt,mirror},xshift=-0.5pt,yshift=-0.6pt](-0.5,-0.6) -- (2.5,-0.6) node[black,midway,yshift=-0.7cm] { \textcolor{darkblue}{$A_1$}};
            \begin{scope}[xshift=3cm]
                \Block[3,midblue!60!white,A,1];
                \draw [thick,midblue,decorate,decoration={brace,amplitude=10pt,mirror},xshift=-0.5pt,yshift=-0.6pt](-0.5,-0.6) -- (2.5,-0.6) node[black,midway,yshift=-0.7cm] { \textcolor{midblue}{$A_2$}};
            \end{scope}
            \begin{scope}[xshift=6cm]
                \Block[3,cyan!60!white,B,1];
                \draw [thick,cyan,decorate,decoration={brace,amplitude=10pt,mirror},xshift=-0.5pt,yshift=-0.6pt](-0.5,-0.6) -- (2.5,-0.6) node[black,midway,yshift=-0.7cm] { \textcolor{cyan}{$A_3$}};
            \end{scope}
            \begin{scope}[xshift=9cm]
                \Block[3,lightgreen!60!white,C,1];
                \draw [thick,lightgreen,decorate,decoration={brace,amplitude=10pt,mirror},xshift=-0.5pt,yshift=-0.6pt](-0.5,-0.6) -- (2.5,-0.6) node[black,midway,yshift=-0.7cm] { \textcolor{lightgreen}{$A_4$}};
            \end{scope}
            \begin{scope}[xshift=12cm]
                \Block[3,midgreen!60!white,D,1];
                \draw [thick,midgreen,decorate,decoration={brace,amplitude=10pt,mirror},xshift=-0.5pt,yshift=-0.6pt](-0.5,-0.6) -- (2.5,-0.6) node[black,midway,yshift=-0.7cm] { \textcolor{midgreen}{$A_5$}};
            \end{scope}
            \node at (15.25,0) {\ldots};
                \begin{scope}[xshift=16.5cm]
                \Block[3,lightpink!60!white,D,1];
                \draw [thick,lightpink,decorate,decoration={brace,amplitude=10pt,mirror},xshift=-0.5pt,yshift=-0.6pt](-0.5,-0.6) -- (2.5,-0.6) node[black,midway,yshift=-0.7cm] { \textcolor{lightpink}{$A_{n-1}$}};
            \end{scope}
            \begin{scope}[xshift=19.5cm]
                \Block[3,darkpurple!60!white,D,1];
                \draw [thick,darkpurple,decorate,decoration={brace,amplitude=10pt,mirror},xshift=-0.5pt,yshift=-0.6pt](-0.5,-0.6) -- (2.5,-0.6) node[black,midway,yshift=-0.7cm] { \textcolor{darkpurple}{$A_{n}$}};
            \end{scope}
            \node at (22.5,0.8) {\huge $\Lambda$};
        \end{tikzpicture}
    \end{center}
    \caption{Representation of an interval $\Lambda$ split into multiple subintervals $\Lambda=A_1 A_2 \ldots A_n$.}
    \label{fig:3}
\end{figure}
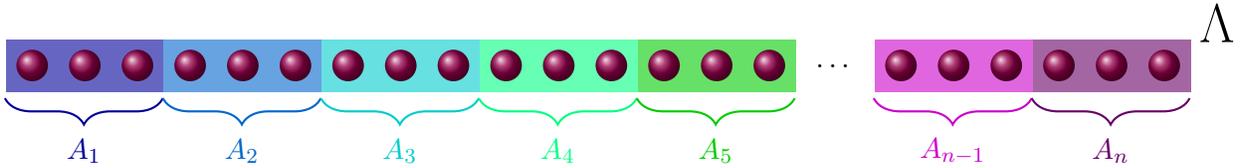

Their size is chosen so that the inverted BS-CMI is sufficiently small, and will be chosen in the next subsection.
For simplicity, let us denote by $\rho_i:= \rho_{A_i} = \tr_{\Lambda \setminus A_i}[\rho]$ and by  $\rho_{i:j}:= \rho_{A_i, A_{i+1}, \ldots , A_{j}}= \tr_{\Lambda \setminus (A_i \cup \ldots \cup A_{j})}[\rho]$.

A standard argument to extend a recovery result from a single recovery map to a chain is to use the contractivity of the recovery map to show that the errors behave additively.
This is not possible in our setting, as the non-trace-preserving maps are not contractive. However, we will be able to overcome this technical difficulty by proving a Lipschitz bound independently of the level of concatenation.
More specifically, in Lemma~\ref{lem:concatLip} we show that the errors only suffer a constant amplification that does not grow exponentially in the length of the recovery chain. 

We denote the individual recovery maps $\cR_i:\mathcal B(\mathcal H_{A_i})\to\mathcal B(\mathcal H_{A_iA_{i+1}})$, $i=1,\ldots,N-1$ as
\begin{equation}\label{eq:recoverymap}
    \cR_i(X)=\rho_{i}^{1/2}(\rho_{i}^{-1/2}\rho_{i:{i+1}}\rho_{i}^{-1/2})^{1/2}\rho_{i}^{-1/2}X\rho_{i}^{-1/2}(\rho_{i}^{-1/2}\rho_{i:{i+1}}\rho_{i}^{-1/2})^{1/2}\rho_{i}^{1/2} \, .
\end{equation}
For simplicity, we will abbreviate the map when adding on a larger chain, and replace $\mathds{1}^{\otimes i-1}\otimes\cR_i$ by $\cR_i$.

\begin{lemma}\label{lem:concatLip}
    Let $1\leq j < N$, be natural numbers and $X\in\mathcal B(\mathcal H_{A_j})$ be positive semidefinite.
    There is a constant independent of the number of concatenated maps that bounds the Lipschitz constant of the concatenated map:
    \begin{align*}
        \left\|\left(\bigcirc_{i=j}^{N-1}\cR_{i}\right)(X)\right\|_1\le \norm{\rho_j^{-1/2}X\rho_j^{-1/2}}_1 \le \|\rho_{j}^{-1}\|\|X\|_1 \, .
    \end{align*}
    This also holds if the $\rho_{i},\rho_{i:{i+1}}$ are not consistent marginals of a fixed global state, as long as they are local positive states and $\tr_j[\rho_{j:{j+1}}]=\rho_{{j+1}}$. 
\end{lemma}
\begin{proof}
    The map $\bigcirc_{i=j}^{N-1}\cR_{i}(\cdot)$ preserves positivity being a concatenation of positive maps. Hence by
    the operator inequality $0 \le X \le \Vert\rho_j^{-1/2} X \rho_j^{-1/2}\Vert \rho_j$, we readily observe that:
    \begin{equation*}
        \norm{\bigcirc_{i=j}^{N-1}\cR_{i}(X)}_1 \le \norm{\rho_j^{-1/2}X\rho_j^{-1/2}}\norm{\bigcirc_{i=j}^{N-1}\cR_{i}(\rho_j)}_1 \, . 
    \end{equation*}
    Further inspection reveals that $\cR_i(\rho_i) = \rho_{i:i + 1}$ and $[\tr_i, \cR_k] = 0$ for $k > i$. Leveraging these observations, we can immediately conclude:
    \begin{equation*}
        \begin{aligned}
            \norm{\bigcirc_{i=j}^{N-1}\cR_{i}(\rho_j)}_1 &= \Tr[\bigcirc_{i=j}^{N-1}\cR_{i}(\rho_j)] = \Tr[\tr_j[\bigcirc_{i=j}^{N-1}\cR_{i}(\rho_j)]] \\
            &= \Tr[\tr_j[\bigcirc_{i=j + 1}^{N-1}\cR_{i}(\rho_{j:j+1})]] \\
            &= \Tr[\bigcirc_{i=j + 1}^{N-1}\cR_{i}(\tr_j[\rho_{j:j+1}])] = \Tr[\bigcirc_{i=j + 1}^{N-1}\cR_{i}(\rho_{j+1})] \, .
        \end{aligned}
    \end{equation*}
    Iterating this process yields $\norm{\bigcirc_{i=j}^{N-1}\cR_{i}(\rho_j)}_1 = \Tr[\rho_N] = 1$. Applying Hölder's inequality and the hierarchy of Schatten p-norms to simplify $\norm{\rho_j^{-1/2}X\rho_j^{-1/2}}$ finally confirms the desired result.
\end{proof}


We now combine Lemma~\ref{lem:concatLip} with Lemma~\ref{lem:singleRecoveryError} to decompose the overall recovery error, obtaining a bound on the recovery error based on entropic quantities. This is the main technical result of the section. 

\begin{theorem}\label{thm:entropicMPO}
    For a multipartite quantum state $\rho_{1:N}$, the recovery error on the chain between the concatenated recovery map and the original state is bounded by
    \begin{equation}\label{eq:recoverycondition}
        \begin{aligned}
           & \left\|\left(\bigcirc_{i=1}^{N-1}\cR_i\right)(\rho_{1})-\rho_{1:N}\right\|_1 \\
           & \hspace{1.2cm} \le\frac{16(N-1)}{\pi}\sup_i \norm{\rho_i^{-1}}\exp(I_\infty(A_1\ldots A_{i-1}:A_i)/2) 
             \widehat{I}_\rho^{\operatorname{rev}}(A_i:A_1\ldots A_{i-2}|A_{i-1})^{1/4} \, .
        \end{aligned}
    \end{equation}
    Note that the recovery map representation $\left(\bigcirc_{i=1}^{N-1}\cR_i\right)(\rho_{1})$ is also a matrix product operator representation with bond dimension $D=\dim(A_i)^3$.
\end{theorem}
\begin{proof}
    We denote the overall recovery error by $\eps$
    \begin{equation}\label{eq:epsilon_error_recovery}
        \eps=\left\|\cR_{N-1}(\ldots\cR_1(\rho_{1})\ldots)-\rho_{1:N}\right\|_1=\left\|\left(\bigcirc_{i=1}^{N-1}\cR_i\right)(\rho_{1})-\rho_{1:N}\right\|_1.
    \end{equation}
    First, we split the total error $\eps$ into the individual errors for each recovery step, where every error is still amplified by the subsequent recovery maps:
    \begin{equation}\label{eq:decomposedPreContractivity}
        \begin{aligned}
             \eps&\le\sum_{i=1}^{N-1}\left\|\left(\bigcirc_{j=i+1}^{N-1}\cR_j\right)(\rho_{1\ldots {i+1}}-\cR_i(\rho_{1\ldots i}))\right\|_1\\
            &\le\sum_{i=1}^{N-1}\left\|\left(\bigcirc_{j=i+1}^{N-1}\cR_j\right)(X_{A_{i+1}}^+)\right\|_1+\left\|\left(\bigcirc_{j=i+1}^{N-1}\cR_j\right)(X_{A_{i+1}}^-)\right\|_1
        \end{aligned}
    \end{equation}
    Here, $X_{i}=\rho_{1:{i+1}}-\cR_i(\rho_{A_1\ldots A_i})$ with decomposition into positive and negative parts $X_{i}=X_{i}^+-X_{i}^-$ and $X_{A_i}^{\pm}=\Tr_{A_1\ldots A_{i-1}}[X_{i}^\pm]$ such that the last norm is taken on systems $A_i\ldots A_N$ only.
    Furthermore, this means that
    \begin{equation*}
        \|X_{A_i}^\pm\|_1=\|X_{i}^\pm\|_1\le\|X_{A_i}\|_1.
    \end{equation*}
    Finally, combining Eqs. \eqref{eq:epsilon_error_recovery} and \eqref{eq:decomposedPreContractivity}, and using Lemma~\ref{lem:concatLip} we can upper bound the error as
    \begin{equation}\label{eq:epsilon}
        \begin{aligned}
            \eps
            &=\left\|\left(\bigcirc_{i=1}^{N-1}\cR_i\right)(\rho_{1})-\rho_{1:N}\right\|_1\\
            &\le2(N-1)\sup_i \|\rho_{i}^{-1}\|\frac8\pi\sqrt{\|\Gamma_i\|}\widehat I_\rho^{\operatorname{rev}}(A_i:A_1\ldots A_{i-2}|A_{i-1})^{1/4} \, ,
        \end{aligned}
    \end{equation}
    where in the last inequality we are using that, by \Cref{lem:singleRecoveryError}, 
    \begin{equation*}
        \norm{X_{A_i}}_1 \leq \frac8\pi\sqrt{\|\Gamma_i\|}\widehat I_\rho^{\operatorname{rev}}(A_i:A_1\ldots A_{i-2}|A_{i-1})^{1/4} \, , 
    \end{equation*}
    and we introduced $\Gamma_i=\rho_{1:{i-1}}^{-1/2}\rho_{1:i}\rho_{1:{i-1}}^{-1/2}$.
    To obtain the final form of the result, let us note that all the quantities involved can be rewritten in terms of divergences.
    The factor $\Gamma_i$ is related to the maximal Rényi mutual information as follows:
    \begin{align*}
        \|\Gamma_i\| &\le \left\|\rho_{i}^{1/2}\right\|^2 \left\|(\rho_{1:{i-1}}\otimes\rho_{i})^{-1/2}\rho_{1:i}(\rho_{1:{i-1}}\otimes\rho_{i})^{-1/2}\right\|\\
        &\le\exp(I_\infty(A_1\ldots A_{i-1}:A_i)) \, .
    \end{align*}
    This concludes the proof.
\end{proof}

We now consider the implications of Theorem~\ref{thm:entropicMPO} for MPO representations of Gibbs states.
Note that due to the expected exponential growth of $\norm{\rho_i^{-1}}$ in Eq. \eqref{eq:recoverycondition} with the size of the subsystem, at least an exponential decay of $\widehat I_\rho^{\operatorname{rev}}(A\mathbin{;}C|B)$ with sufficient rate, for any adjacent three systems $A, B, C$, is needed to employ the above result.
In the case of thermal states, however, Theorem \ref{theorem:superexponential-decay-BS-CMI} guarantees a superexponential rate, which dominates over the other prefactors.  

\begin{corollary}\label{cor:gibbsMPO}
    For an $n$-site marginal of a Gibbs state in one dimension with translation-invariant interaction $\Phi$ and a given accuracy $\eps$, there is an MPO representation of bond dimension
    \begin{equation} \label{eq:bonddim}
        D=\exp(2\log(d)\cC_1\frac{\log(n/\eps)+\cC_2}{\log(\log(n/\eps))}),
    \end{equation}
    where $\cC_1,\cC_2$ are constants depending on $J$, $R$, and $\beta$.
    
    The MPO representation is given by the above construction, i.e.,
    \begin{equation*}
        \left\|\left(\bigcirc_{i=1}^{N-1}\cR_i\right)(\rho_{1})-\rho_{1:N}\right\|\le\eps
    \end{equation*}
    choosing the $A_i$ as consecutive regions of at least $l$ spins
    \begin{equation*}
         |A_i|=l\ge \cC_1\frac{\log(n/\eps)+\cC_2}{\log(\log(n/\eps))} \, . 
    \end{equation*}
\end{corollary}
\begin{remark}
    Due to the uniformity of the result in \cref{theorem:superexponential-decay-BS-CMI}, this result applies equally to a local Gibbs state, i.e. the Gibbs state of the Hamiltonian on $n$ sites, as to the marginal of a Gibbs state on a larger or even infinite system.
\end{remark}
\begin{proof}[Proof of \cref{cor:gibbsMPO}]
    Here, we follow a similar procedure to that of \cite[Corollary 3.3]{Fawzi.2023}. First, recall that \cref{theorem:superexponential-decay-BS-CMI} provides the needed decay for the BS-CMI in the case of thermal states.
    Furthermore, in this setting, by Lemma~\ref{lemma:norm-estimates-of-functions-of-Gibbs-states}, the growth of $\norm{\rho_i^{-1}}$ is exponential in $|A_i|$.
    Finally, the maximal mutual information obeys an area law for thermal states \cite[Theorem 2]{Scalet.2021}, namely
    \begin{equation*}
         I_\infty(A:B)\le\cC
    \end{equation*}
    for any adjacent intervals $A, B$ and a constant $\cC$ that only depends on $\beta$, $J$ and $R$.
    
    We now proceed by partitioning our system into subsystems of size $|A_i|=l$ to be determined later.
    Combining all these results we achieve the desired accuracy if
    \begin{align*}
        \eps^4\le\cC_1\frac{n}{l}\frac{\cC_2^{\lfloor l/2\rfloor+1}}{(\lfloor\lfloor l/2\rfloor/R\rfloor+1)!} \, .
    \end{align*}
    Note that the factor $\exp(\alpha(|A_i| + |A_{i+1}|))$ for every $i$ is included in $\cC_2^{\lfloor l/2\rfloor+1}$.
    Using Stirling's approximation we find that a sufficient condition for this is given by
    \begin{equation*}
         \log(\frac{n\cC_1}{\eps^4})\le-\left(\frac l2+1\right)\log(\cC_2)+\left(\frac{l}{2R}-\frac1R\right)\left(\log(\frac{l}{2R}-\frac1R)-1\right) \, .
    \end{equation*}
    While this inequality cannot be inverted analytically, a further relaxation based on a dual description of the convex function $x\mapsto x\log(x)$ yields the sufficient condition 
    \begin{equation*}
         2R\frac{\log(\frac{n\cC_1\cC_2^2}{\eps^4})+a}{\log(a)-R\log(\cC_2)}+1\le l 
    \end{equation*}
    for any $a>0$. Choosing $a=\log(n/\eps)\mathcal{C}_2^R$ and again combining all constants the claim follows.
\end{proof}

\begin{remark}
    \Cref{cor:gibbsMPO} can be tailored to accommodate exponentially-decaying interactions, thereby yielding an MPO approximation for a Gibbs state in this context, which to the best of our knowledge is the first of its kind. The bond dimension is slightly worse in that case though, as it arises from the decay of the BS-CMI, which is only exponential in $|B|$ in this case, as opposed to the superexponential behaviour in the finite-range case. This is the content of \cref{cor:gibbsMPO_shortrange}.
\end{remark}

\subsection{Reconstructing the positive MPO from local tomography}\label{sec:tomography}
Since Theorem \ref{thm:entropicMPO} provides a way to reconstruct an entire state $\rho$ from the smaller marginals, it also offers a reconstruction of classical representations of $\rho$, in the form of a MPO, from tomographic data on small regions. What is required from the tomographic data is to produce a classical representation of the marginals with small tomography errors, which can be done with standard results \cite{Haah.2017,Wright.2016}. To prove that our final MPO reconstruction of $\rho$ is accurate, we thus need to show that the reconstruction process, which happens at the level of classical post-processing, is stable to the small errors in the local tomography. Since in our setting the Hamiltonian is translation invariant, the cost of learning a state from measurements is limited to the cost of learning the state in a single region (although the precision might depend on the final size of the reconstruction).

Our starting point is a measurement scheme for general quantum states that we will apply to the marginals of the Gibbs state.
Since we do not aim to give constants in our runtime explicitly, any polynomial time measurement scheme is sufficient for our purposes.
Sample optimal measurement schemes for arbitrary quantum states with quadratic sample complexity in the inverse error and dimension have been developed in \cite{Haah.2017,Wright.2016}.
The situation regarding the classical runtime of these schemes is however unclear.
Nevertheless, the following result is an immediate consequence of measuring each coefficient of the density matrix individually to sufficient precision and using the equivalence of matrix norms up to dimensional factors.
\begin{lemma}[consequence of {\cite{Haah.2017,Wright.2016}}] \label{le:tomography}
    There exists a measurement scheme that produces a classical representation of a quantum state $\hat{\rho}$ such that with probability $\ge1-c$ 
    \begin{equation}
        \norm{\rho-\hat \rho }_1 \le \delta \, ,
    \end{equation}
    with several samples and classical runtime bounded by
    \begin{align*}
        n_s&\le C_s \poly\left(d,\frac{1}{\delta}\right)\log(1/c) \, ,\\
        t&\le C_t \poly\left(d,\frac1\delta\right)\log(1/c) \, .
    \end{align*}
\end{lemma}
With this, we can now state the cost of local tomography that we require to produce the accurate MPO representation of the whole state from the previous section. The proof essentially shows how the errors in the marginals from \cref{le:tomography} propagate in the final error of the MPO approximation, which allows us to estimate the final costs of samples and classical runtime.

\begin{theorem}\label{thm:learnMPO}
    Let us consider a $n$-site marginal of a Gibbs state in one dimension with translation-invariant interaction, and define the approximate reconstruction map as in Eq. \eqref{eq:recoverymap}, 
    \begin{equation}\label{eq:hat_recovery_condition}
        \hat{\cR}_i(X)=(\hat \rho_{i})^{1/2}((\hat \rho_{i})^{-1/2}\hat \rho_{i:i+1}(\hat \rho_{i})^{-1/2})^{1/2}(\hat \rho_{i})^{-1/2}X(\hat \rho_{i})^{-1/2}((\hat \rho_{i})^{-1/2}\hat \rho_{i:i+1}(\hat \rho_{i})^{-1/2})^{1/2}(\hat \rho_{i})^{1/2} \, ,
    \end{equation}
    where each $\hat{\rho}_i$ is an approximation of the corresponding marginal of $\rho$. Then, for a given accuracy $\eps$, the MPO representation $\left(\bigcirc_{i=1}^{N-1}\hat{\cR}_i\right)(\hat \rho_{1})$ has bond dimension as in Eq. \eqref{eq:bonddim} and is such that, with probability $\ge1-c$,
    \begin{equation*}
        \left\|\left(\bigcirc_{i=1}^{N-1}\hat{\cR}_i\right)(\hat \rho_{1})-\rho_{1:N}\right\|_1\le\eps \, . 
    \end{equation*}
    The sample complexity and classical post-processing time for finding this MPO are $\textup{poly}\left(n /\varepsilon \right) \log(1/c)$.
\end{theorem}

\begin{proof}
    We assume the repartitioned subsystems $A_1,\ldots,A_N$ each of size $|A_i|=l$ as of \cref{cor:gibbsMPO}.
    We need to show that the accuracy of the reconstruction map $\cR_i(X)$ is robust to small errors in the estimated marginals $\rho_{i},\rho_{i:i+1}$.
    We repeatedly apply those approximate reconstructions, and the error induced by each is shown in Eq. \eqref{eq:epsilon}.
    
    To construct a good approximation to $\left(\bigcirc_{i=1}^{N-1}\cR_i\right)(\rho_1)$ we require estimates of the marginals $\hat \rho_{i}$ and $\hat \rho_{i:i+1}$, which we take to be such that with probability $\ge1-c$
    \begin{equation}\label{eq:tomography}
        \norm{ \rho_{i:i+1}-\hat \rho_{i:i+1}}_1 \le \delta,
    \end{equation}
    for which we need $\textup{poly}\left( d^l,\delta^{-1} \right)\log(N/c)$ samples. Then, for consistency, we take the single-region marginals as the reductions $\hat \rho_{{i+1}}=\tr_{i}[\hat \rho_{i:i+1}]$.
    To account for the possibility of failure in the tomography scheme, we choose a confidence bound of $c/N$ for each marginal such that, by a union bound, all marginals fulfil Eq.~\eqref{eq:tomography}.
    
    To obtain an accurate approximation to the recovery map, we need our estimates to have well-behaved inverses.
    Let us note that according to \cref{lemma:norm-estimates-of-functions-of-Gibbs-states} 
    \begin{equation*}
        \norm{\rho_{i:i+1}^{-1}}\le \cC e^{\alpha l} \, .
    \end{equation*}
    Thereby, explicitly requiring
    \begin{equation}\label{eq:deltaCon1}
        \delta\le\frac1{2\cC} e^{-\alpha l} \le \frac12 \left\|\rho_{i:i+1}^{-1}\right\|^{-1} \, ,
    \end{equation}
    we obtain the equivalent bounds for the approximate marginals
    \begin{equation*}
        \left\|\rho_{i:i+1}^{-1}\right\|,\left\|\rho_{i}^{-1}\right\|, \left\|(\hat \rho)_{i:i+1}^{-1}\right\|,\left\|(\hat \rho)_{i}^{-1}\right\|\le 2\cC e^{\alpha l} \, .
    \end{equation*}
    With this and the definition of $\hat{\cR}_i$ from  Eq. \eqref{eq:hat_recovery_condition},  our target MPO approximation is $\left(\bigcirc_{i=1}^{N-1}\hat{\cR}_i\right)({\hat \rho}{}_{1})$. 
    To control the approximation error, we repeatedly apply the triangle inequality so that
    \begin{align*}
        \norm{\left(\bigcirc_{i=1}^{N-1} \hat{\cR}_i \right)({\hat \rho}{}_{1}) - \left(\bigcirc_{i=1}^{N-1}\cR_i\right)(\rho_{1}) }_1 \le \sum_{j=1}^{N-1}& \norm{\left(\bigcirc_{i=j+1}^{N-1}\cR_i\right)\left(\cR_j-\hat{\cR}_j \right)\left(\bigcirc_{i=1}^{j-1}\hat{\cR}_i\right)({\hat \rho}_{1})}_1 \\&+\norm{\left(\bigcirc_{i=1}^{N-1}\cR_i\right)(\hat \rho_{1}-\rho_{1})}_1.
    \end{align*}
    
    Applying Lemma \ref{lem:concatLip} to each term in the sum, and Eq. \eqref{eq:tomography} we have that 
    \begin{align*}
        \norm{\left(\bigcirc_{i=1}^{N-1}\hat{\cR}_i\right)({\hat \rho}{}_{1}) - \left(\bigcirc_{i=1}^{N-1}\cR_i\right)(\rho_{1}) }_1 \le \sum_{j=1}^{N-1}& \norm{\rho_{j+1}^{-1}} \norm{\left(\cR_j-\hat{\cR}_j \right)\left(\bigcirc_{i=1}^{j-1}\hat{\cR}_i\right)(\hat \rho_{1})}_1 \\&+2\delta \norm{\rho_{1}^{-1}}.
    \end{align*}
    To bound the summands $\norm{\left(\cR_j-\cR'_j \right)(X)}_1$, we require the following two continuity bounds for the square root and its inverse, as proven in \cite{Cavaretta.2003},
    \begin{align*}
        \norm{\sqrt{\rho}-\sqrt{{\hat \rho}{}}} &\le \max \{\norm{\rho^{-1}},\norm{({\hat \rho}{})^{-1}}\}\cdot \sqrt{8}\norm{\rho-{\hat \rho}{}}, \\
        \norm{\rho^{-1/2}-({\hat \rho}{})^{-1/2}} &\le \max\{\norm{\rho^{-1}},\norm{({\hat \rho}{})^{-1}}\} \cdot \norm{\rho^{-1/2}}\norm{({\hat \rho}{})^{-1/2}}
        \cdot \sqrt{8}\norm{\rho-{\hat \rho}{}}.
    \end{align*}
    Repeated applications of these and the triangle inequality, and using that $\norm{\cdot} \le \norm{\cdot}_1$ then yield
    \begin{align*}
        \norm{\left(\cR_j-\hat{\cR}_j \right)(X)}_1 \le C \delta \norm{X}_1 \left( \norm {\rho_{j}{}^{-1}}\norm{(\hat \rho_{j})^{-1}}\right)^{D}
    \end{align*}
    for some numerical constants $C,D$.
    Again, notice that by Lemma \ref{lem:concatLip},
    \begin{equation*}
        \norm{\left(\bigcirc_{i=1}^{j-1}\hat{\cR}_i\right)(\hat \rho_{1})}_1 \le 1.
    \end{equation*}
    Putting everything together, and introducing new numerical constants $C', D'$, we thus obtain
    \begin{equation*}
        \norm{\left(\bigcirc_{i=1}^{N-1} \hat{\cR}_i \right)(\hat \rho_{1}) - \left(\bigcirc_{i=1}^{N-1}\cR_i\right)(\rho_{1}) }_1 \le C' N \delta d^{\alpha D' l}. 
    \end{equation*}
    This means (recalling Eq.~\eqref{eq:deltaCon1}) that if we set 
    \begin{equation*}
         \delta=\min\left\{\frac{\varepsilon d^{-\alpha D' l}}{C'N},\frac{1}{2\cC}e^{-\alpha l}\right\}
    \end{equation*}
    by \cref{cor:gibbsMPO} and the triangle inequality,
    \begin{equation*}
        \norm{\left(\bigcirc_{i=1}^{N-1} \hat{\cR}_i \right)(\hat \rho_{1}) - \rho_{1:N} }_1 \le 2 \varepsilon.
    \end{equation*}
    Given the choice of $\delta$, the sample complexity is bounded by $\textup{poly}\left(d^l,N /\varepsilon\right)\log(1/c)$.
    The time complexity, apart from the contribution due to the tomography itself, consists of elementary matrix functions and multiplications which also run in time $\textup{poly}\left(d^l,N\right)$.
    Noting the sublogarithmic dependence of $l$ on $n/\eps$ in \cref{cor:gibbsMPO} and that $N\le n$, the result follows. 
\end{proof}
\begin{remark}\label{rem:subpolyTomo}
    Compared to the previous section which gave a \emph{subpolynomial} dependence in system size and inverse error, we only achieve polynomial dependence here.
    For the inverse error, this is due to the polynomial dependence in the tomography scheme and cannot be avoided.
    However, for the system size, the situation is slightly different.
    Instead of explicitly giving all $N$ channels it can be noted that for \emph{infinite} translation-invariant systems, they are all identical and only measuring one of the marginals and reconstructing one of the channels can be done in subpolynomial time in $n$.
    Similarly, for \emph{finite} translation-invariant only the sample complexity can be improved to subpolynomial in $n$ by measuring all subsystems simultaneously for each sample.
    However, writing out the MPO for the $n$ systems takes linear time in $n$ by definition again in both cases.
\end{remark}

\subsection{Efficient estimation of the global purity}\label{sec:MPOpurity}
Beyond MPO reconstructions of the entire state, the results on the factorization of the purity from Section \ref{sec:purity} allow us to efficiently estimate the purity of the global state $\rho_{1:N}$ from local purity estimations.

For completeness, we now reproduce a simplified version of the approximation scheme of the purity of 1D Gibbs states in \cite{Vermersch.2023}, where a more detailed analysis of the polynomial scaling of resources and prefactors can be found. Our main contribution is Proposition \ref{th:purity}, which constitutes the missing ingredient for a fully rigorous polynomial bound on the sample complexity of their estimation algorithm. Instead, in \cite{Vermersch.2023}, the exponential decay in Eq. \eqref{eq:approximate_factorization} was only studied numerically. To simplify the discussion, we do not consider the randomized measurement toolbox \cite{Elben.2022}, which potentially yields a slightly more favourable scaling and more amenable classical post-processing, and instead use the result from Lemma \ref{le:tomography}. 

\begin{theorem}\label{thm:estimation_purity}
    Under the conditions of Proposition \ref{th:purity}, there is an algorithm that outputs an estimate $\hat P_2(\rho_{1:N})$ such that
    \begin{equation} \label{eq:purityerror}
        \left \vert \frac{\hat P_2(\rho_{1:N})}{\Tr[\rho_{1:N}^2]}- 1 \right \vert \le \varepsilon \, ,
    \end{equation}
    with a number of samples and classical post-processing cost given by $\textup{poly}(n/\varepsilon)$.
\end{theorem}
\begin{proof}
    Consider adjacent marginals of the state $\rho_{1:N}$ on regions of size $l$ and adjacent pairs $\rho_{i:i+1}$. An estimate for the purity is given by
    \begin{equation*}
        P_2(\rho_{1:N})=\frac{\prod_{j=1}^{N-1}\Tr_{j:j+1}[\rho^2_{j:j+1}]}{\prod_{j=2}^{N-1}\Tr_{j}[\rho^2_{j}]} \, .
    \end{equation*}
    Then, from an iterated application of Proposition \ref{th:purity}, and the straightforward estimate for scalars $x,y$,
    \begin{equation*}
        |xy -1| \leq |x| |y-1| + |x-1| \, ,
    \end{equation*}
    we have that (see also Lemma 5 in \cite{Vermersch.2023})
    \begin{align*}
        \left \vert \frac{P_2(\rho_{1:N})}{ \Tr[\rho_{1:N}^2] }-1\right \vert \le \left (1+ \kappa e^{-\alpha l}\right)^N-1 \, .
    \end{align*}
    Therefore, with the choice $l=\mathcal{O}\left(\log{N/\varepsilon}\right)$ we achieve error $\varepsilon$. 
    
    The numbers $\Tr_{j:j+1}[\rho^2_{j:j+1}]$ and $\Tr_{j}[\rho^2_{j}]$
    can be obtained to multiplicative error $\delta$ from $\textup{poly}(d^l, \delta^{-1})$ local measurements and the scheme in Lemma \ref{le:tomography}, since 
    \begin{equation*}
        \frac{\Tr_{j}[\rho^2_{j}]}{\Tr_{j}[\hat \rho^2_{j}]} -1 \le \Tr_{j}[\hat \rho^2_{j}]^{-1}\norm{\rho_{j} - \hat \rho_{j}}_1 \le d^{l} \delta \, ,
    \end{equation*}
    and thus our estimation of the global purity is 
    \begin{equation*}
        \hat P_2(\rho_{1:N})=\frac{\prod_{j=1}^{N-1}\Tr_{j:j+1}[\hat \rho^2_{j:j+1}]}{\prod_{j=2}^{N-1}\Tr_{j}[\hat \rho^2_{j}]} \, .
    \end{equation*}
    With a similar argument as above, this is such that
    \begin{equation*}
        \left \vert \frac{\hat P_2(\rho_{1:N})}{ P_2(\rho_{1:N})} - 1 \right \vert \le \left (1+ d^{2l} \delta \right)^{2N}-1 \, .
    \end{equation*}
    Therefore, choosing $d^{l} \delta=e^{-\Omega(l)}= \mathcal{O}(1/ \text{poly}(N/\varepsilon) )$ we again obtain the approximation error $\varepsilon$. The result follows from the estimates on the sample and computational cost in Lemma \ref{le:tomography}, and noticing that $n\ge N$. 
\end{proof}

Notice that it is also possible to estimate the purity via the estimation of the MPO representation in \cref{thm:learnMPO}, since
\begin{equation*}
    \left \vert \Tr[\left(\bigcirc_{i=1}^{N-1}\hat{\mathcal R}_i\right)(\hat \rho_{1})^2]-\Tr[\rho_{1:N}^2] \right \vert \le \norm{\left(\bigcirc_{i=1}^{N-1}\hat{\mathcal R}_i\right)(\hat \rho_{1})-\rho_{1:N}}_1.
\end{equation*}
This, however, is a much worse estimate than Eq. \eqref{eq:purityerror}, since the purity is typically an exponentially small number. Thus, additive approximations (as opposed to multiplicative ones) are much less meaningful.

\begin{acknowledgements}
AC, PG and AR-d-A acknowledge the support of the Deutsche Forschungsgemeinschaft (DFG, German Research Foundation) - Project-ID 470903074 - TRR 352. AC, PG and AR-d-A also acknowledge funding by the Federal Ministry of Education and Research (BMBF) and the Baden-Württemberg Ministry of Science as part of the Excellence Strategy of the German Federal and State Governments. 
SOS acknowledges support from the UK Engineering and Physical Sciences Research Council (EPSRC) under grant number EP/W524141/1.
AMA acknowledges support from the Spanish Agencia Estatal de Investigaci\'on through the grants ``IFT Centro de Excelencia Severo Ochoa CEX2020-001007-S" and ``Ram\'on y Cajal RyC2021-031610-I", financed by MCIN/AEI/10.13039/501100011033 and the European Union NextGenerationEU/PRTR. This project was funded within the QuantERA II Programme which has received funding from the EU’s H2020 research and innovation programme under the GA No 101017733. AR-d-A acknowledges financial support by the Ministry for Digital Transformation and of Civil Service of the Spanish Government through the QUANTUM ENIA project call – Quantum Spain project, and by the European Union through the Recovery, Transformation and Resilience Plan – NextGenerationEU within the framework of the Digital Spain 2026 Agenda.
\end{acknowledgements}

\vspace{0.5cm}

\noindent \emph{Data Availability:}
Data sharing is not applicable to this article as no datasets were generated or analysed during the current study. \\
\emph{Conflict of interests:}
The authors have no competing interests to declare that are relevant to the content of this article.

\bibliographystyle{abbrv}
\bibliography{references}

\appendix

\section{Proofs of approximate factorization of Gibbs states in 1D}\label{sec:proofsApproxTens}
This appendix is devoted to the proofs of the new results appearing in Section \ref{sec:approximate-factorisation}, in connection to the approximate factorisation of Gibbs states in 1D.

\subsection{Proof of \texorpdfstring{\Cref{Lemma:Modified_BoundingPartialTraceInverses}}{Lemma~\ref*{Lemma:Modified_BoundingPartialTraceInverses}}}

Let us recall that we want to prove the following inequality:
\begin{equation}\label{eq:aux_inverse_norm_contraction_app}
    \norm{\tr_{AC}[\rho^{A}\otimes \rho^{C} E_{AB,C} \, E_{A, B}]^{-1}} \,   \leq \, \mathcal{C} \, .
\end{equation}
Before providing its proof, as a technical tool we introduce the following measure of locality
\begin{equation*}
    \norm{T}_{I}:=\inf\{\norm{T - P}: P\in\cA_I\}
\end{equation*}
and the following Lemma, whose proof can be found in \cite{Bluhm.2022}.
\begin{lemma}{\cite[Lemma 4.1]{Bluhm.2022}}\label{lem:interval-norms}
    Let $T,T'$ be two observables. Then
    \begin{equation*}
        \|TT'\|_I\le 2(\|T\|_I\|T'\|+\|T\|\|T'\|_I)
    \end{equation*}
    and if $T$ is positive definite
    \begin{equation*}
        \|T^{-1}\|_{I}\le 2\|T^{-1}\|^2\|T\|_I.
    \end{equation*}
    For $T$ an observable on a bipartite system $AB$, $\rho_A$ any state on $A$, $I\subset AB$, and $I'=I\cap B$, we also have
    \begin{equation*}
        \|\tr_A[\rho_A T]\|_{I'}\le\|T\|_I
    \end{equation*}
\end{lemma}
\begin{proof}
    The proof of the last point has been omitted in the reference, but follows straightforwardly from the contractiveness of the operator norm. By compactness, let $P\in\cA_I$ be an operator such that $\|T\|_I=\|T-P\|$.
    Then,
    \begin{equation*}
        \|\tr_A[\rho_AT]\|_{I'}\le\|\tr_A[\rho_A T]-\tr_A[\rho_A P]\|\le\|T-P\|=\|T\|_I.
    \end{equation*}
\end{proof}
\begin{proof}[Proof of \Cref{Lemma:Modified_BoundingPartialTraceInverses}]
The proof of this lemma closely follows the arguments of the third point of \cite[Corollary 4.4]{Bluhm.2022}, however, due to subtle differences in the statement we reprove it here. The original strategy used that the operator $Q$ is almost local in a small interval (after a partial trace and identification of neighbouring intervals), a preliminary that we replace with $Q$ being approximately supported on two small \emph{but distant} intervals. First, note that
\begin{align*}
    \tr_{AC}[\rho^{A}\otimes \rho^{C} Q] & =  \tr_{AC}[\rho^{A}\otimes \rho^{C} e^{-H_{ABC}} e^{H_{A} +H_{B}+ H_{C}}] \\
    & = e^{\frac{1}{2}H_{B}} \tr_{AC}[\rho^{A}\otimes \rho^{C} F_{A,B} F_{AB,C} F_{AB,C}^* F_{A,B}^*   ]e^{-\frac{1}{2}H_{B}}\\
    & = e^{\frac{1}{2}H_{B}} \tr_{AC}[\rho^{A}\otimes \rho^{C}F]e^{-\frac{1}{2}H_{B}}\, ,
\end{align*}
with $F_{X,Y} = e^{-\frac{1}{2}H_{XY}} e^{\frac{1}{2}(H_{X}+H_Y)} $ and $F = F_{A,B} F_{AB,C} F_{AB,C}^* F_{A,B}^* \ge 0$. By \cref{prop:estimates_expansionals}, we know that each factor $F_{X,Y},F_{X,Y}^{-1}$ is bounded in norm by $\mathcal{G}$ (for $\Phi/2$), and thus $\norm{F},\norm{F^{-1}} \le \cG^4$.

Let us define the intervals $I_n$ as neighbourhoods of the cuts $AB$ and $BC$ and $I'_n$ as $I_n$ reduced to $B$. More precisely if $l_\cdot$ and $u_\cdot$ denote the lower and upper boundaries of an interval, we define $I_n = \{u_A - n, \hdots ,  l_B + n\} \cup \{u_{B} - n, \hdots,  l_{C} + n\}$. Our objective now is to minimize the impact of $F$ by approximating it solely in terms of local operators on these sets. Subsequently, we will demonstrate that by performing a partial trace, these operators are further reduced to act non-trivially solely on $I'_n = B \cap I_n$.

Using \cref{lem:interval-norms} and \cref{prop:estimates_expansionals} we can derive 
\begin{align*}
    \norm{F}_{I_n} &= \norm{F_{A,B} F_{AB,C} F_{AB,C}^* F_{A,B}^*}_{I_n}\\
    &\le 4\cG^3 \left(\|F_{A,B}\|_{I_n}+\| F_{AB,C}\|_{I_n}+\| F_{AB,C}^*\|_{I_n}+\| F_{A,B}^*\|_{I_n}\right)\\
    &\le16\cG^3 \frac{\cG^n}{(\lfloor n/R\rfloor+1)!}.
\end{align*}
Using \cref{lem:interval-norms} and the fact that the map $\tr_{AC}[\rho^{A}\otimes\rho^{C}\cdot]$ is positive unital we have
\begin{align*}
    \|F'\|_{I'_n}&\le\|F\|_{I_n}\le 16\cG^3 \frac{\cG^n}{(\lfloor n/R\rfloor+1)!}\\
    \|F'^{-1}\|&\le\|F^{-1}\| \le \cG^4,
\end{align*}
where $F'=\tr_{AC}[\rho^{A}\otimes\rho^{C}F]$. Employing \cref{lem:interval-norms} again
\begin{equation*}
    \|F'^{-1}\|_{I'_n}\le2\|F'^{-1}\|^2\|F'\|_{I'_n}\le 32\cG^{11}\frac{\cG^n}{(\lfloor n/R\rfloor+1)!}.
\end{equation*}

The next step is completely analogous to \cite[Lemma 4.2]{Bluhm.2022} with the only deviation being that \cite[Proposition 3.2]{Bluhm.2022}, involving uniform estimates for the complex time evolution of any observable, does not apply directly.

Let us denote by $P_n$ a collection of $I'_n$-local operators fulfilling
$\|F'^{-1}-P_n\|=\|F'^{-1}\|_{I'_n}$ and $S_1 = P_1$, $S_n = P_n - P_{n-1}$.
This implies the bounds
\begin{equation*}
    \|S_1\|\le \|F'^{-1}\|_{I'_1}+\|F'^{-1}\|\qquad\textrm{ and }\qquad \|S_n\|\le\|F'^{-1}\|_{I'_n}+\|F'^{-1}\|_{I'_{n-1}} \, .
\end{equation*}
Note that for $n\ge\lfloor|B|/2\rfloor+1$, $P_n = F'^{-1}$ and thereby $S_{n+1}=0$.
We start with a decomposition
\begin{equation*}
    F'^{-1}=\sum_{n=1}^{\lfloor|B|/2\rfloor+1} S_n \, ,
\end{equation*}
where we stop the local decomposition once the two distinct intervals start overlapping. Our goal is to bound
\begin{equation*}
    e^{\frac12H_{B}}F'^{-1}e^{-\frac12H_{B}}
\end{equation*}
by individually bounding the norms of the imaginary time-evolved $S_n$. To this end, let us note that there exists an orthogonal decomposition for all $n\le\lfloor|B|/2\rfloor$
\begin{equation*}
    S_n=\sum_{i=1}^{d^{4n}} S^l_{n,i}\otimes S^r_{n,i} \, . 
\end{equation*}
Such a decomposition further possesses the property that $\|S^l_{n,i}\otimes S^r_{n, i}\| \leq d^{2n} \|S_n\|$, which is derived from the orthogonality of the $S^l_{n,i}\otimes S^r_{n, i}$ and the inequality
\begin{equation*}
    \|\cdot\|\le\|\cdot\|_2\le\sqrt{d'}\|\cdot\| \, .
\end{equation*}
Here, $d'$ is the dimension of the underlying Hilbert space. We can now use that the time evolution is a group automorphism on each term individually and that we have \cite[Proposition 3.2]{Bluhm.2020} allowing us to estimate
\begin{align*}
    \left\|e^{\frac12H_{B}}S^l_{n,i}\otimes S^r_{n,i}e^{-\frac12H_{B}}\right\|&\le\left\|e^{\frac12H_{B}} S^l_{n,i}e^{-\frac12H_{B}}\right\|\left\|e^{\frac12H_{B}}S^r_{n,i}e^{-\frac12H_{B}}\right\|\\
    &\le \cG^{2n} d^{2n}\|S_n\|
\end{align*}
At last for $n=\lfloor|B|/2\rfloor+1$, we simply regard $R_n$ as an $B$-local operator, so the directly applying \cite[Proposition 3.2]{Bluhm.2020} yields
\begin{equation*}
    \left\|e^{\frac12H_{B}}S_ne^{-\frac12H_{B}}\right\|\le\cG^{|B|}\|S_n\| \, . 
\end{equation*}

Putting everything together, we have
\begin{align*}
    \left\|e^{\frac12H_{B}}F'^{-1}e^{-\frac12H_{B}}\right\|&\le 32 d^6\cC^2(\cG^{12}+\cG^4)+\left(32\cG^{11}\sum_{n=2}^{\lfloor|B/2|\rfloor+1} \cG^{2n} d^{6n}\frac{\cG^n}{(\lfloor n/R\rfloor+1)!} \right) 
\end{align*}
and the statement follows by combining constants.
\end{proof}

\subsection{Proof of \texorpdfstring{\Cref{Lemma:DPI_mixing_condition}}{Lemma~\ref*{Lemma:DPI_mixing_condition}}}

Next, we prove that, for a Gibbs state $\rho^{A'ABCC'}$ on a finite chain $A'ABCC'$, the following holds.
\begin{equation}\label{eq:DPI_mixing_condition_app}
   \left\| \rho_{AB} \, \rho_{B}^{-1}  \rho_{BC} \, \rho_{ABC}^{-1} - \1 \right\|_{\infty} < \mathcal{C}^4  \, \left\| \rho_{A'AB} \, \rho_{B}^{-1} \rho_{BCC'} \, \rho_{A'ABCC'}^{-1} - \1 \right\|_{\infty} \, ,  
\end{equation}

\begin{proof}
We absorb $\beta$ in the interaction, so the argument will be made for $\beta=1$. Let us rewrite
\begin{align*} 
& \rho_{AB} \, \rho_{B}^{-1} \rho_{BC} \, \rho_{ABC}^{-1} \\[1.5mm]
& \hspace{2cm} = \operatorname{tr}_{A'C'}\left[ \rho_{A'AB} \, \rho_{B}^{-1} \rho_{BCC'} \right] \, \operatorname{tr}_{A'C'} \left[\rho_{A'ABCC'} \right]^{-1}  \\[1.5mm]
& \hspace{2cm} =  \operatorname{tr}_{A'C'} \left[ \rho_{A'AB} \, \rho_{B}^{-1} \rho_{BCC'} \, \rho_{A'ABCC'}^{-1} \, \rho_{A'ABCC'} \right]  \, \operatorname{tr}_{A'C'} \left[\rho_{A'ABCC'}\right]^{-1}   \\[1.5mm]
& \hspace{2cm} = \operatorname{tr}_{A'C'} \left[ \rho_{A'AB} \, \rho_{B}^{-1} \rho_{BCC'} \, \rho_{A'ABCC'}^{-1}  \, e^{-H_{A'ABCC'}} \right] \, \operatorname{tr}_{A'C'} \left[e^{-H_{A'ABCC'}}\right]^{-1}  \\[1.5mm]
& \hspace{2cm} = \operatorname{tr}_{A'C'} \left[ \rho_{A'AB} \, \rho_{B}^{-1} \rho_{BCC'} \, \rho_{A'ABCC'}^{-1} \, e^{-H_{A'ABCC'}} \, e^{H_{ABC}} \right] \, \operatorname{tr}_{A'C'} \left[e^{-H_{A'ABCC'}} \, e^{H_{ABC}}\right]^{-1} \, .
\end{align*}
Recalling that
\[ E_{A'ABC, C'}:= e^{-H_{A'ABCC'}} \, e^{H_{A'ABC} + H_{C'}}\; , \quad E_{A',ABC}:= e^{-H_{A'ABC}} \, e^{ H_{A'} + H_{ABC} } \, , \]
then, from the above identity, we deduce that
\begin{align*}
 & \rho_{AB} \, \rho_{B}^{-1} \rho_{BC} \, \rho_{ABC}^{-1} \\[1.5mm]
 & = \operatorname{tr}_{A'C'} \left[ \rho_{A'AB} \, \rho_{B}^{-1} \rho_{BCC'} \, \rho_{A'ABCC'}^{-1} \, E_{A'ABC, C'} \, E_{A',ABC} \, e^{-H_{A'}-H_{C'}} \right] \, \operatorname{tr}_{A'C'} \left[E_{A'ABC, C'} \, E_{A',ABC} \, e^{-H_{A'}-H_{C'}} \right]^{-1}  \\[1.5mm]
& = \operatorname{tr}_{A'C'} \left[ \rho^{A'}\otimes\rho^{C'} \, \rho_{A'AB} \, \rho_{B}^{-1} \rho_{BCC'} \, \rho_{A'ABCC'}^{-1} \, E_{A'ABC, C'} \, E_{A',ABC} \right] \, \operatorname{tr}_{A'C'} \left[\rho^{A'}\otimes\rho^{C'} \,E_{A'ABC, C'} \, E_{A',ABC}  \right]^{-1} 
\end{align*}
Thus
\begin{align*}
    & \rho_{AB} \, \rho_{B}^{-1} \rho_{BC} \, \rho_{ABC}^{-1} - \1 \\[1.5mm]
  &  \hspace{2.3cm} = \operatorname{tr}_{A'C'} \left[ \rho^{A'}\otimes\rho^{C'} \, Q \,E_{A'ABC, C'} \, E_{A',ABC} \right] \, \operatorname{tr}_{A'C'} \left[\rho^{A'}\otimes\rho^{C'}\, \,E_{A'ABC, C'} \, E_{A',ABC}  \right]^{-1}\,, 
\end{align*}
where
\[ Q := \rho_{A'AB} \, \rho_{B}^{-1} \rho_{BCC'} \, \rho_{A'ABCC'}^{-1} - \1\,. \]
By  \cref{eq:contractiveExpectation}, we have
$$\norm{\operatorname{tr}_{A'C'} \left[ \rho^{A'}\otimes\rho^{C'} \, Q \, E_{A'ABC, C'} \, E_{A',ABC} \right]} \leq \norm{E_{A'ABC, C'}} \norm{E_{A',ABC}} \norm{Q} \leq \mathcal{C}^2 \norm{Q} \, ,$$
and for the term with the inverse we have an upper bound of $ \mathcal{C}^2 $ by \Cref{Lemma:Modified_BoundingPartialTraceInverses}. Therefore,
\[ \left\| \rho_{AB} \, \rho_{B}^{-1} \rho_{BC} \, \rho_{ABC}^{-1} - \1 \right\| \leq \mathcal{C}^4  \,\| Q \| \,. \]
\end{proof}

\section{Extension of results to exponentially-decaying interactions}\label{sec:shortrange}
In this appendix, we extend the main results of this manuscript to the framework of exponentially-decaying interactions. All results presented in \cref{sec:approximate-factorisation} admit extensions to this setting, albeit with some small modifications in the estimates. Here, we collect the generalisations of all these technical results in the context of exponentially-decaying interactions and either refer to a source for proof or give the proof here. Next, we use them to generalise the main results of this manuscript (\cref{theorem:superexponential-decay-BS-CMI}, \cref{thm:estimation_purity} and \cref{thm:learnMPO}) to exponentially-decaying interactions. To increase readability, we will only highlight those parts of the proofs that differ from the finite-range counterpart. 

Consider an interaction on $\Sigma \subseteq \mathbb{Z}$. To define exponentially-decaying interactions, we introduce for each $\lambda>0$ the following notation
\begin{align*}\label{eq:norm_interaction}
\Omega_n &:= \sup_{x\in\Sigma}\sum\{\norm{\Phi(X)}:X\ni x, \diam(X)\ge n\}\\
    \| \Phi\|_{\lambda}&:=  \sum_{n\ge 0} \Omega_n e^{\lambda n} \in [0, \infty]\,.
\end{align*}
We say that  $\Phi$ is \textit{exponentially decaying}, if there exists $\lambda > 0$ such that $\norm{\Phi}_{\lambda} < \infty$. Note that there is a subtle difference in the definition of $\norm{\Phi}_\lambda$ in the papers from which we extract the results presented and utilised below. The above definition is the one from \cite{Garcia.2023} which is conformal with $\vertiii{\cdot}_\lambda$ from \cite{Bluhm.2024} and \cite{Capel.2024} in the sense that, $\vertiii{\Phi}_{\lambda} \le \norm{\Phi}_\lambda \le \vertiii{\Phi}_{\lambda - \varepsilon}$ for all $\varepsilon > 0$.\footnote{For proof of the relation, the reader can consult \cite[Section 9.1]{Capel.2024}.} Hence by requiring $\norm{\Phi}_\lambda < \infty$ we immediately get that $\vertiii{\Phi}_\lambda < \infty$ is satisfied for the same $\lambda$ and we further can replace appearances of the latter with the first.
Note further that the {finite range} interactions are {exponentially decaying}. Given a finite interval $\Lambda \Subset \Sigma$ split into $X$ and $Y$, the expansional of a Hamiltonian on $X,Y$ is defined analogously to Eq. \eqref{def:expantionals}, and the estimates of \cref{prop:estimates_expansionals} extend similarly (see \cite{Garcia.2023}). The important difference is that (ii) presents exponential decay with $\ell$, rather than superexponential and that the result only holds up to some critical temperature $\crit>0$. Here we recall the simplified formulation of \cite[Lemma 31]{Capel.2024}.

\begin{proposition}[{\cite[Theorem 3.1]{Garcia.2023}}]\label{prop:estimates_expansionals_shortrange}
    Let $\Phi$ be an exponentially-decaying interaction which is further translation invariant and let $\beta<\crit$. Then the following hold: 
    \begin{enumerate}
        \item[(i)] There is an absolute constant $\tilde{\mathcal{G}}>1$ depending only on $\lambda$ and $\beta$ such that, for any finite interval $\Lambda = X Y \Subset \mathbb{Z}$ split into two subintervals $X$ and $Y$, we have:
        \begin{equation*}
            \norm{E_{X,Y}} \, , \,\norm{E_{X,Y}^{-1}} \,\, \leq \,\, \tilde{\mathcal{G}} \, .
        \end{equation*}
        \item[(ii)] There are positive constants $\mathcal{K}, \alpha >0$ depending on $\lambda$ and $\beta$ such that if we add two intervals $\widetilde{X}$ and $\widetilde{Y}$ adjacent to $X$ and $Y$, respectively, so that we get a larger interval $\widetilde{X}XY\widetilde{Y}$, then
        \begin{equation*}
            \left\| E_{X,Y}^{-1} - E^{-1}_{\widetilde{X}X,Y\widetilde{Y}} \right\|, \left\| E_{X,Y} - E_{\widetilde{X}X,Y\widetilde{Y}} \right\| \, \leq \, \mathcal{K} e^{-\alpha \ell} \,.
        \end{equation*}
        for any $\ell \in \mathbb{N}$ such that $\ell \, \leq \, |X| \, , \, |Y|$. 
     \end{enumerate}
\end{proposition}
Next, note that \cref{lem:interval-norms} is independent of the range of the interactions. Note further that \cref{prop:BoundingPartialTraceInverses} relies on estimates for expansionals and certain contractions thereof which allows us to conclude its proof for exponentially-decaying interactions directly from \cref{lem:interval-norms} and point (i) of \cref{prop:estimates_expansionals_shortrange}. The statement can be found in \cref{prop:BoundingPartialTraceInverses_shortrange}. By similar arguments, we conclude \cref{Lemma:Modified_BoundingPartialTraceInverses} in the exponentially-decaying setting.

\begin{proposition}\label{prop:BoundingPartialTraceInverses_shortrange}
    Under the conditions above, there is an absolute constant $\tilde{\mathcal{C}}>1$ depending only on $\lambda$ and $\beta$ such that
    \begin{align}
       &  \big\| \tr_{B}[\rho^{B}Q] \big\|, \big\| \tr_{B}[\rho^{B}Q]^{-1} \big\| \, \leq \, \tilde{\mathcal{C}} \, \,  , \quad \quad Q \in \{ E_{B,C}^{*} \, , \, E_{B,C}\, , \, E_{A,B}^{*} \, , \, E_{A,B} \} \, , \label{equa:BoundingPartialTraceInverses1_sr}\\[2mm]
      &   \label{equa:BoundingPartialTraceInverses2_sr} \big\| \tr_{AB}[\rho^{AB}Q] \big\| \, , \, \big\| \tr_{AB}[\rho^{AB}Q]^{-1} \big\| \,   \leq \, \tilde{\mathcal{C}} \, \,  , \quad \quad Q \in \{ E_{A,B}^{* \, -1} \, , \, E_{A,B}^{\, -1} \} \, ,\\[2mm]
      &  \label{equa:BoundingPartialTraceInverses3_sr} \big\| \tr_{B}\big[\rho^{B}E_{A,B}^{*} E_{AB,C}^{*}\big] \big\| \, , \, \big\| \tr_{B}\big[\rho^{B}E_{A,B}^{*} E_{AB,C}^{*}\big]^{-1} \big\| \, \leq \, \tilde{\mathcal{C}} \, ,\\[2mm]
      &  \label{equa:BoundingPartialTraceInverses4_sr}   \norm{\tr_{AC}[\rho^{A}\otimes \rho^{C} E_{AB,C} \, E_{A,BC} ]^{-1}}  \, \leq \, \tilde{\mathcal{C}} \, .
    \end{align}
\end{proposition}

The next result constitutes a generalisation of \cite[Theorem 5.1]{Bluhm.2022} to the setting of exponentially-decaying interactions. Its proof completely follows that of its finite-range counterpart, and thus we only provide a very brief sketch of it, highlighting the differences.

\begin{theorem}\label{lemma:approximate-factorization-of-Gibbs-state_shortrange}
   With $\Phi$ on $\Z$ an exponentially-decaying interaction that is further translation-invariant and with $\beta<\crit$, there are positive constants $\tilde{K}, \tilde{\alpha}>0$ depending only on $\lambda$ and $\beta$ such that for every $\Lambda \Subset \Z$ split into three subintervals $\Lambda = ABC$, where $B$ shields $A$ from $C$, for its local Gibbs state $\rho^{\Lambda} =: \rho_{ABC}$ it holds that
    \begin{equation}\label{eq:approximate-factorization-of-Gibbs-state_shortrange}
        \norm{\rho_{ABC} \rho_{BC}^{-1} \rho_B \rho_{AB}^{-1} - \1} \le \tilde{K} e^{-\tilde{\alpha} |B|} \, .
    \end{equation}
\end{theorem}
\begin{proof}[Sketch of the proof]
    From \cite[Eq. (19)]{Bluhm.2022}, we can rewrite
    \begin{equation*}
        \rho_{ABC} \rho_{BC}^{-1} \rho_B \rho_{AB}^{-1} = E_{A,BC} \tr_A[ E_{A,BC} \rho^A]^{-1} \tr_A[ \tilde{E}_{A,BC} \rho^A] \tilde{E}_{A,BC}^{-1} \, ,
    \end{equation*}
    with 
    \begin{equation*}
        \tilde{E}_{A,BC} = \tr_C[ e^{-H_{ABC}}] \tr_C[e^{-H_{BC}}]^{-1} e^{H_A} \, .
    \end{equation*}
    Now, the conclusion of the proof is a combination of two statements:
    \begin{enumerate}
        \item Statement: The four factors in the right-hand side of the previous expression and their inverses are uniformly bounded. This is a consequence of \cref{prop:estimates_expansionals_shortrange} and \cref{prop:BoundingPartialTraceInverses_shortrange}. The universal constant is given by $\tilde{\mathcal{G}} \, \tilde{\mathcal{C}}^2 $.
        \item Statement: The following inequalities hold,
        \begin{align*}
         \norm{ \tilde{E}_{A,BC} - E_{A,BC} } & \leq \tilde{K}' e^{-\tilde{\alpha}' |B|} \, , \\
         \norm{ \tr_A[ \tilde{E}_{A,BC} \rho^A] - \tr_A[ E_{A,BC} \rho^A] } & \leq \tilde{K}' e^{-\tilde{\alpha}' |B|} \, .
        \end{align*}
        The second inequality follows from the first one by contractivity. For the first one, we split the terms in the same way as in the proof of \cite{Bluhm.2022}, obtaining $\tilde{K}'= 4 \, \mathcal{C} \, \mathcal{K}  \, e^{- \alpha}$ and $\tilde{\alpha}'= \alpha/2 $. 
    \end{enumerate}
    We conclude by taking $ \tilde{\alpha}=\tilde{\alpha}'$ and $\tilde{K}=  8 \, \tilde{\mathcal{C}}^7  \tilde{\mathcal{G}}^3 \mathcal{K}   e^{-\alpha}$.
\end{proof}

Next, one can immediately conclude \cref{Lemma:DPI_mixing_condition} for exponentially-decaying interactions as of Eq. \eqref{equa:BoundingPartialTraceInverses4_sr}. This, jointly with \cref{lemma:approximate-factorization-of-Gibbs-state_shortrange}, allows us to deduce a generalisation of \cref{lemma:approximate-factorization-of-Gibbs-state_shortrange} to the case where the Gibbs state is given on $\Lambda = A'ABCC'$ and but we only compare its marginals on $ABC, AB, B$ and $BC$, respectively. Namely, we can show that, in this context, the following inequality holds:
\begin{equation}\label{eq:approximate-factorization-of-Gibbs-state_longerchain_shortrange}
    \norm{\rho_{ABC} \rho_{BC}^{-1} \rho_B \rho_{AB}^{-1} - \1} \le \tilde{C}^4 \tilde{K} e^{-\tilde{\alpha} |B|} \, .
\end{equation}

To conclude the preliminaries for exponentially-decaying interactions, let us recall that the exponential decay of correlations and local indistinguishability hold at any positive temperature for the considered setting. The proofs of these results can be found in \cite[Theorem 26]{Capel.2024} and  \cite[Theorem 28]{Capel.2024}, respectively, as a consequence of \cite[Theorem 4.4]{Garcia.2023}.

\begin{proposition}[Decay of correlations and local indistinguishability]\label{prop:decaycorr_localind_shortrange}
    For $\Phi$ an exponentially-decaying, translation invariant interaction over $\Z$, the following inequalities hold for universal constants $c,c', \gamma, \gamma' >0$ depending only on $\lambda$ and $\beta$:
      \begin{align*}
        \left|\Tr_{ABC}[\rho^{ABC} O_A \, O_C ] -  \Tr_{ABC}[\rho^{ABC} O_A  ]  \Tr_{ABC} [\rho^{ABC} O_C  ]  \right| & \leq \norm{O_A} \norm{O_C} \, c \, e^{- \alpha |B|} \,  ,\\
       \left| \Tr_{ABC}[\rho^{ABC} O_A] -  \Tr_{AB}[\rho^{AB} O_A]  \right| & \leq \norm{O_A} \, c' \, e^{-\gamma' |B|} \, , \\[1mm]
      \left| \Tr_{ABC}[\rho^{ABC} O_C] -  \Tr_{BC}[\rho^{BC} O_C]  \right| & \leq \norm{O_C} \, c' \, e^{-\gamma' |B|} \, .
    \end{align*}
\end{proposition}

\subsection{Exponential decay of the BS-CMI for exponentially-decaying interactions}

First, note that most estimates from \cref{lemma:norm-estimates-of-functions-of-Gibbs-states} follow straightforwardly in the context of exponentially-decaying interactions, and only the last one differs. For that, we have as before
\begin{align*}
    \norm{\rho_B^{-1}}\norm{\rho_B} & =  \left\|(\rho^B)^{-1}\tr_{AC}[E_{A,BC}E_{B,C}\rho^A\rho^C]^{-1}\right\|\left\|\rho^B\tr_{AC}[E_{A,BC}E_{B,C}\rho^A\rho^C]\right\| \, ,
\end{align*}
with the difference that now 
\begin{align*}
    \norm{ (\rho^B)^{-1}}\norm{\rho^B} \leq e^{2 \norm{H_B}} \leq e^{2 |B| \beta\norm{\Phi}_{\lambda} } \, .
\end{align*}
Having this in mind, we can show the following bounds for the decays of the various versions of the BS-CMI by replacing the respective bounds in the proof of \cref{theorem:superexponential-decay-BS-CMI}.

\begin{theorem}\label{theorem:superexponential-decay-BS-CMI_shortrange}
    Let $\Phi$ an exponentially-decaying, translation invariant interaction over $\Z$ and $\Lambda \Subset \Z$ any interval split into consecutive parts $\Lambda = A'ABCC'$, with $A'$ and $C'$ possibly empty. Then for the marginal on $ABC$ of its local Gibbs state $\tr_{A'C'}[\rho^{\Lambda}] = \rho_{ABC}$ there exist a critical temperature $\tcrit$ universal constants $c, \alpha , \gamma>0$ independent of $\Lambda$ and only dependent on $\lambda$ and $\beta$ such that 
    \begin{equation}
        \widehat I^{x}_{\rho} (A\mathbin{;}C|B) \le \tilde{c} e^{\tilde{\alpha}|A|} e^{ -\gamma |B| } \quad x \in \{\operatorname{os}, \operatorname{ts}, \operatorname{rev}\} \, . 
    \end{equation}
\end{theorem}
\begin{proof}
    The bounds themselves are completely analogous to the proof in \cref{theorem:superexponential-decay-BS-CMI}, replacing the superexponentially decaying function $\eps(|B|)$ by $\tilde K e^{-\tilde\alpha|B|}$ from \cref{lemma:approximate-factorization-of-Gibbs-state_shortrange}.
    Now whenever $\tilde \alpha>2\beta\|\Phi\|_{\lambda}$, we obtain a positive decay rate $\gamma=\tilde \alpha-2\beta\|\Phi\|_{\lambda}$.
    Since $\tilde \alpha=\lambda-2\beta\Omega_0$ (see \cite[Section 2.3.2]{Garcia.2023}) we can find $\tcrit$ such that $\gamma>0$ for $\beta<\tcrit$.
\end{proof}


\subsection{Efficient estimation of the global purity for exponentially-decaying interactions}

The objective of this subsection is to develop the findings about the efficient evaluation of the overall purity of a Gibbs state on a translation-invariant spin chain in the context of interactions that decay exponentially. We will apply the necessary technical results to this context to achieve this. However, most results have similar proofs to the finite-range case, so we will skip those parts for simplicity and only highlight the differences.

First, we can prove the following proposition: the generalisation of \cref{th:purity}.

\begin{proposition}\label{th:purity_shortrange}
    Let $\Phi$ be an exponentially-decaying, translation-invariant interaction. Then, for $\beta<\crit$ there exist positive constants $\tilde c_p, \tilde{\alpha}_p$ depending only on $\lambda$ and $\beta$ such that, for every $\Lambda \Subset \mathbb{Z}$ split as $\Lambda = ABC$ and for $\rho_\Lambda := \rho^\Lambda$ the Gibbs state on $\Lambda$, the following holds.
    \begin{equation}\label{eq:approximate_factorization_shortrange}
        \left| \frac{\Tr_{AB}[\rho_{AB}^2]\Tr_{BC}[\rho_{BC}^2]}{\Tr_\Lambda[\rho_\Lambda^2] \Tr_B[\rho_B^2] } -1\right| \leq \tilde{c}_p e^{- \tilde{\alpha}_p |B|} \, .
    \end{equation}
\end{proposition}
\begin{proof}[Sketch of the proof]
    The proof is analogous to that of \cref{th:purity}. Let us recall that we can estimate the left-hand side of Eq. \eqref{eq:approximate_factorization_shortrange} by:
    \begin{equation*}
        \left|\widetilde{\lambda}_{ABC}^{-1} - 1 \right| + \left| \widetilde{\lambda}_{ABC}^{-1} \right| \left|  \chi_{ABC} - 1 \right| \, ,
    \end{equation*}
    where the terms are defined completely analogously to the finite-range case, but we keep the same notation for simplicity. The estimates of $ \left|\widetilde{\lambda}_{ABC}^{-1} - 1 \right|$ and $\left| \widetilde{\lambda}_{ABC}^{-1} \right| $ were shown in \cite[Lemma 6.1]{Bluhm.2024} to hold exactly in the same way as for finite range. In the finite range case we used a combination of estimates on Araki's expansionals, decay of correlations and local indistinguishability to conclude the decay of $\left|\chi_{ABC} - 1 \right|$. All these properties also hold for exponentially-decaying interactions as a consequence of \cref{prop:estimates_expansionals_shortrange} and \cref{prop:decaycorr_localind_shortrange}, hence we can conclude the proof. 
\end{proof}

Next, note that the global purity of a Gibbs state split into many regions can be estimated efficiently from this proposition in the same way as in \cref{thm:estimation_purity}. Since the proof is the same as for \cref{thm:estimation_purity}, and follows from \cite{Vermersch.2023}, we omit it here.

\begin{theorem}\label{thm:estimation_purity_shortrange}
    Under the conditions of \cref{th:purity_shortrange}, there is an algorithm that outputs an estimate $\hat P_2(\rho_{1:N})$ such that
    \begin{equation} \label{eq:purityerror_shortrange}
        \left \vert \frac{\hat P_2(\rho_{1:N})}{\Tr[\rho_{1:N}^2]}- 1 \right \vert \le \varepsilon,
    \end{equation}
    with a number of samples and classical post-processing cost given by $\textup{poly}(n/\varepsilon)$.
\end{theorem}

\subsection{Learning of Gibbs states via MPO approximations for exponentially-decaying interactions}

In analogy with the previous subsection, the recovery map employed in the reconstruction of the Gibbs state from its marginals for the MPO approximation is defined in the same way for exponentially-decaying interactions as for the finite-range ones. Therefore, the estimate in \cref{thm:entropicMPO} for the MPO approximation is still valid, and we omit its explicit formulation in this context for simplicity. The only change occurs in the analogue of \cref{cor:gibbsMPO}, since the bond dimension is inherited from the decay of the BS-CMI, and thus it is expected to be slightly worse in this case. Indeed, as a consequence of \cref{theorem:superexponential-decay-BS-CMI_shortrange}, we obtain the following. 

\begin{corollary}\label{cor:gibbsMPO_shortrange}
    For a $n$-site marginal of a Gibbs state in one dimension with translation-invariant, exponentially-decaying interaction $\Phi$ and a given accuracy $\eps$, there is an MPO representation of bond dimension
    \begin{equation} \label{eq:bonddim_shortrange}
        D=\exp(2\log(d)\tilde{\cC}_1\log(n/\eps)),
    \end{equation}
    where $\tilde{\cC}_1$ is a constant depending on $\lambda$ and $\beta$ and $\tilde{\alpha}>0$ is obtained from \cref{theorem:superexponential-decay-BS-CMI_shortrange}.  The MPO representation is given by
    \begin{equation*}
        \left\|\left(\bigcirc_{i=1}^{N-1}\cR_i\right)(\rho_{1})-\rho_{1:N}\right\|\le\eps
    \end{equation*}
    choosing the $A_i$ as consecutive regions of at least $l$ spins
    \begin{equation*}
         |A_i|=l\ge \cC_1{\log(n/\eps)} \, . 
    \end{equation*}
\end{corollary}

\section{Lifting the upper bound on the DPI of the BS-entropy to channels}
Note that just by restriction of the original Hilbert space to the common support of $X$ and $Y$, one can obtain the same result as \cref{theorem:upper-bound-DPI-of-BS-entropy} for the case that $X$ and $Y$ are not full-rank but share the same kernel. In that case, the inverses are replaced by Moore-Penrose pseudo-inverses. With this at hand, we can now prove the upper bound on the DPI for the BS-entropy for quantum channels instead of conditional expectations.
\begin{corollary}\label{corollary:extending-upper-bound-bs-entropy-to-channels}
    For $X, Y \in \cB(\cH_A)$ having the same support and $\cT:\cB(\cH_A) \to \cB(\cH_B)$ a quantum channel, we find that
    \begin{equation}
        \begin{aligned}
            \widehat{D}(X \Vert Y) - \widehat{D}(\cT(X) \Vert \cT(Y)) &\le \norm{X^{-1/2} Y X^{-1/2}} \norm{X}_1  \norm{\cT(X)^{1/2}} \norm{\cT(X)^{-1/2}}\\
            &\hphantom{\le} \cdot \norm{\cT(Y)^{-1} \cT(X)} \norm{XY^{-1}\cT^*(\cT(Y)\cT(X)^{-1}) - \1}
        \end{aligned}
    \end{equation}
    where inverses are replaced by Moore-Penrose pseudo-inverses in case the operators are not full-rank.
\end{corollary}
\begin{proof}
    Following similar arguments as in \cite{Wilde.2018, Bluhm.2020}, we note that by Stinespring's dilation theorem, we can write every quantum channel as a composition:
    \begin{equation}\label{eq:steinsprings-dilation-theorem}
        \cT(\cdot) = \tr_{C}[V \cdot V^*]
    \end{equation}
    where $V: \cH_A \to \cH_B \otimes \cH_C$ is an isometry and $\cH_C$ an auxiliary space. Using the isometric invariance of the BS-entropy and its additivity under tensor products, we get
    \begin{align*}
        \widehat D(X \Vert Y) - \widehat{D}(\cT(X) \Vert \cT(Y)) = \widehat{D}(X_V \Vert Y_V) - \widehat{D}(\pi_C \otimes \tr_C[X_V] \Vert \pi_C \otimes \tr_C[Y_V]) \, ,
    \end{align*}
    where we set $X_V = V X V^*$ and $Y_V = V Y V^*$. Both of the latter operators agree in their support, hence with the considerations brought forward at the beginning of this appendix we utilise \cref{theorem:upper-bound-DPI-of-BS-entropy} to obtain 
    \begin{align*}
         \widehat D(X \Vert Y) - \widehat{D}(\cT(X) \Vert \cT(Y)) &\le \norm{X_V^{-1/2} Y_V X_V^{-1/2}} \norm{X_V}_1 \norm{\cT(X)^{1/2}}\norm{\cT(X)^{-1/2}}\\
         &\hphantom{\le} \cdot \norm{\cT(Y)^{-1} \cT(X)}\norm{X_V Y_V^{-1} (\cT(Y) \cT(X)^{-1})\otimes \1_C - \1}
    \end{align*}
    where we already cancelled constants and used isometric invariance of the Schatten p-norms and Eq. \eqref{eq:steinsprings-dilation-theorem} to simplify the expression. Using again isometric invariance on the remaining terms and identifying $\cT^*$ confirms the claim. 
\end{proof}

\end{document}